\documentclass[letter, 11pt]{article}

\usepackage{natbib}
\usepackage{amssymb,amsmath,xcolor,graphicx,xspace,colortbl,rotating} %
\usepackage[raggedrightboxes]{ragged2e} 
\usepackage{textcomp}
\usepackage{appendix}  
\usepackage{bm}  
\usepackage{boxedminipage}  
\usepackage{color}  
\usepackage{endnotes}  

\usepackage{ragged2e}  
\usepackage[onehalfspacing]{setspace}  
\usepackage{tabulary}  
\usepackage{varioref}  
\usepackage{wrapfig}  

\graphicspath{{Dynamic_walarsian eq7_graphics/}{Dynamic_walarsian eq7_tcache/}{Dynamic_walarsian eq7_gcache/}}
\DeclareGraphicsExtensions{.pdf,.eps,.ps,.png,.jpg,.jpeg}
\usepackage[margin=1in]{geometry}
\usepackage {bm}
\usepackage {tabulary}
\DeclareGraphicsExtensions {.pdf,.svg,.eps,.ps,.png,.jpg,.jpeg}
\newtheorem {theorem}{Theorem}

\newtheorem {assumption}{Assumption}

\newtheorem {definition}{Definition}

\newtheorem {lemma}{Lemma}

\newtheorem {proposition}{Proposition}
\newtheorem {remark}{Remark}

\newenvironment {proof}[1][Proof]{\noindent \textbf {#1.} }{\ \rule {0.5em}{0.5em}}

\begin{document}

\title{General equilibrium in a heterogeneous-agent incomplete-market economy with many consumption goods and a risk-free bond}

\author{Bar Light\protect\footnote{Graduate School of Business,
Stanford University, Stanford, CA 94305, USA. e-mail: \textsf{barl@stanford.edu}}~ ~}
\maketitle

\thispagestyle{empty}

\noindent \noindent \textsc{Abstract}:
\begin{quote}
We study a pure-exchange incomplete-market economy with heterogeneous agents. In each period, the agents choose how much to save (i.e., invest in a risk-free bond), how much to consume, and which bundle of goods to consume while their endowments are fluctuating. We focus on a competitive stationary equilibrium (CSE) in which the wealth distribution is invariant, the agents maximize their expected discounted utility, and both the prices of consumption goods and the interest rate are market-clearing. Our main contribution is to extend some general equilibrium results to an incomplete-market Bewley-type economy with many consumption goods.  Under mild conditions on the agents' preferences, we show that the aggregate demand for goods depends only on their relative prices and that the aggregate demand for savings is homogeneous of degree in prices, and we prove the existence of a CSE. When the agents' preferences can be represented by a CES (constant elasticity of substitution) utility function with an elasticity of substitution that is higher than or equal to one, we prove that the CSE is unique. Under the same preferences, we show that a higher inequality of endowments does not change the equilibrium prices of goods, and decreases the equilibrium interest rate. Our results shed light on the impact of market incompleteness on the properties of general equilibrium models.

\end{quote}


\smallskip \noindent \emph{}

\newpage

\section{Introduction}
We extend the classic Arrow-Debreu model to a dynamic incomplete- market general equilibrium model with a continuum of agents, in which each agent has an individual state that corresponds to his wealth level. There is an infinite number of periods and in each period agents participate in a pure-exchange Arrow-Debreu model (as in the seminal paper by \cite{arrow1954existence}). Each agent has a different wealth level and different preferences over consumption bundles, and thus the agents are heterogeneous in the static pure-exchange Arrow-Debreu model. In each period, given the agents' wealth levels and their preferences over consumption bundles, the agents decide how much to spend on a bundle of goods to be consumed in that period, which bundle of goods to consume, and how much to save for future consumption. In the tradition of the Bewley models (see \cite{ljungqvist2012recursive} for a textbook treatment), the markets are incomplete. The agents face uninsurable idiosyncratic risk and can transfer assets from one period to another only by saving in a risk-free bond.  In each period, the agents receive a random endowment vector. We assume that all random shocks are idiosyncratic, ruling out aggregate random shocks that are common to all agents. As in \cite{arrow1954existence}, the agents are price takers, that is, the agents take the prices of the consumption goods and the risk-free bond's rate of return as given. We focus on a pure-exchange economy without production, and hence we extend  Huggett's model \citep{huggett1993risk} to a setting with many consumption goods.\protect\footnote{
For similar models with one consumption good see \cite{lucas1980equilibrium}, \cite{geanakoplos2014inflationary}, and \cite{hu2019unique}. In \cite{bewley1986stationary} there are multiple consumption goods but the interest rate is fixed and is assumed to be $0$ (see also \cite{karatzas1994construction}). In our paper the interest rate is determined in equilibrium as it is in  \cite{huggett1993risk}.} In Huggett's model there is only one consumption good. In contrast, in the model presented in this paper, there are many consumption goods and each good has a price, which is a typical feature of many economies of interest. Thus, in addition to the standard inter-temporal consumption-savings decision, the agents also make a static decision of how to allocate their spending between the different consumption goods. As discussed below, this leads to several complications in the analysis because the aggregate demand for consumption goods and the aggregate demand for savings are coupled through the consumption goods' prices and the interest rate. 

The solution concept that we study in this paper is competitive stationary equilibrium. A competitive stationary equilibrium (CSE) consists of a wealth distribution, prices of goods, an interest rate, savings policy functions, and demand functions for goods, such that: (i) given the prices of the goods and the interest rate, the agents choose a savings policy function and a demand function for goods in order to maximize their expected discounted utility; (ii) the wealth distribution induced by the agents' decisions is invariant; (iii) the prices of the goods and the interest rate are market-clearing, i.e., for each good, the aggregate supply of that good equals the aggregate demand for that good, and the aggregate supply of savings equals the aggregate demand for savings. We note that   stationary equilibrium is a popular solution concept in Bewley type models.\footnote{Computation of a non-stationary equilibrium is typically infeasible. The notion of stationary equilibrium  is conceptually similar to the notion of mean field equilibrium. In Section 3.4 we compare the solution concept used in this paper - competitive stationary equilibrium - with mean field equilibrium.}  

 In this paper we mainly focus on the theoretical properties of CSE. Despite the recent popularity of Bewley-type models,\footnote{Bewley models feature rich heterogeneity and are widely used to study many economic phenomena (see \cite{heathcote2009}, \cite{de2015quantitative}, and \cite{benhabib2016skewed}, for surveys). These include wealth distribution \citep{benhabib2015wealth}, monetary transmission mechanisms \citep{kaplan2016monetary}, aggregate demand \citep{auclert2018inequality}, and many more.} general theoretical results such as the existence of an equilibrium, the uniqueness of an equilibrium, and analytical comparative statics results, have remained limited. Previous existence results in Bewley-type models rely on the  intersection of the supply and demand curves which can be used in an economy with one asset (see \cite{accikgoz2018existence} and \cite{zhu2020existence}); or on casting the heterogeneous-agent macro model as a discrete time mean field model (see \cite{acemoglu2012} and \cite{LW2018}). These methods cannot be used in our setting or more generally in an heterogeneous-agent incomplete-market setting with more than one asset or more than one consumption good, features typical of many models that are studied the applied literature. 
 The paper's main contribution is to establish some general equilibrium results in an incomplete-market Bewley-type model with many consumption goods. We believe that our results and methods can be extended in the future to study  other Bewley-type models with more than one asset  where, to the best of our knowledge,   there are no general results that guarantee the existence or uniqueness of a stationary equilibrium. We note that proving the existence of a competitive stationary equilibrium in this setting is challenging, as we discuss below. 
 Our model, which combines a Bewley-type model with the classic  Arrow-Debreu model, can be used to study the relationship between the consumption goods' prices and other important economic variables such as the risk-free rate and the wealth distribution\footnote{
There is a vast literature on asset pricing and wealth inequality in models other  than Bewley-type models (for example, \cite{judd2003asset}, \cite{blume2006if}, \cite{krueger2010market}, and \cite{kubler2015life}, just to name a few).} in the popular framework of heterogeneous-agent models. For example, how does an increase in income inequality influence the relative prices? We provide a first result in this direction in Section 3.3.

We now explain our contributions in more detail.

\textbf{Existence.} One of our main theoretical contributions is to provide a proof for the existence of a CSE  under fairly general conditions on the agents' preferences (see Theorem \ref{Theorem exist}). In a Bewley model setting, as we discussed above, previous existence results assume that there is one consumption good (e.g., \cite{accikgoz2018existence}) or that the interest rate is fixed and is not determined in equilibrium (e.g., \cite{bewley1986stationary}). These assumptions simplify the analysis of Bewley models considerably as they decouple the aggregate demand for savings and the aggregate demand for consumption goods. Without this decoupling, it is not clear that even the basic properties of the excess demand function such as Warlas' law and homogeneity hold. In addition, proving the needed properties of the excess demand function that imply the existence of a CSE is challenging.  
Nonetheless, we show that some well-known results that apply to the static Arrow-Debreu model also hold in the incomplete-market Bewley-type model that we study. In Proposition \ref{Prop homogenous of } we show that the excess demand for savings is homogeneous of degree one in the goods' prices while the excess demand for each consumption good is homogeneous of degree zero. This result implies that the aggregate demand for goods and the aggregate supply of goods depend only on their relative prices, similarly to the static Arrow-Debreu model, and thus the competitive stationary equilibrium depends only on the relative prices of goods. We also show that the excess demand satisfies Walras' law and some boundary conditions that are crucial in order to establish the existence of a CSE. Using these properties of the excess demand function and the properties of the savings policy function, we establish the existence of a CSE. 
Welfare theorems (e.g., \cite{arrow1951extension}) that apply to the static Arrow-Debreu model do not hold in our setting because of market incompleteness.\protect\footnote{
See \cite{davila2012constrained}, \cite{shanker2017existence}, \cite{nuno2018social}, and \cite{park2018constrained} for a study of welfare maximization in Bewley models.}

\textbf{Uniqueness.} We prove the uniqueness of a CSE for the special case that the agents' preferences over bundles can be represented by a CES (constant elasticity of substitution) utility function with an elasticity of substitution that is equal to or higher than one (see Theorem \ref{Theorem uniq}). This assumption on the agents' preferences implies that the consumption goods are gross substitutes, i.e., the demand for each consumption good increases with the prices of the other consumption goods. We note that the standard argument for proving the uniqueness of an equilibrium in the static Arrow-Debreu model cannot be applied in our setting because of the coupling between the aggregate demand for consumption goods and the aggregate demand for savings. The aggregate demand for consumption goods is not necessarily increasing with the interest rate and the aggregate demand for savings is not necessarily increasing with the prices of consumption goods. Thus, the excess demand function does not necessarily satisfy the gross substitutes property. We overcome this difficulty by using the properties of the CES utility function, the specific economy that we study, and the results in \cite{light2017uniqueness}.

\textbf{Comparative statics.} Our main result regarding the wealth distribution's influence on the prices of goods and on the interest rate is Theorem \ref{Theorem wealth ineq}. We prove that if the agents' preferences over bundles can be represented by a CES utility function with an elasticity of substitution that is higher than one, then an increase in the risk of the random future endowments (in the sense of the convex stochastic order) changes the CSE in the following way: the interest rate decreases, and the prices of goods do not change. In the classic Arrow-Debreu model the result that the prices of goods do not change when the wealth inequality is higher is intuitive because the demand for each good is linear in wealth. Thus, the demand for each good does not change when the wealth inequality is higher. In our setting, under a CES utility function, the marginal propensity to consume is decreasing, so the demand for each good is concave in wealth. An increase in the risk of the random future endowments \textit{increases} the aggregate savings because of the precautionary savings effect. Thus, the aggregate demand for each good decreases. At the same time, a decrease in the interest rate \textit{decreases} the aggregate savings because of the substitution effect. It turns out that in general equilibrium, the precautionary savings effect and the substitution effect offset each other exactly and the prices of goods do not change.

The rest of the paper is organized as follows. Section 2 presents the model. In Section 2.1 we define the CSE. In Section 3 we present the main results of this paper. In Section 3.1 we establish the existence of a CSE. In Section 3.2 we provide conditions that ensure the uniqueness of a CSE. In Section 3.3 we discuss how the wealth distribution influences the prices of goods and the interest rate. In Section 3.4 we compare the current paper to recent work on mean field games. In Section 3.5 we extend the model to ex-ante heterogeneous agents. In Section 4 we provide final remarks, followed by an Appendix containing proofs.

\section{The model}
There is a continuum of agents of measure $1$. Every agent has an
individual state. We assume for now that the agents are ex-ante identical. In Section 3.5 we extend the model to ex-ante heterogeneous agents. There are $n$ goods where $n \geq 2$. Let $Y_{i}$ be a random variable that describes the evolution of good $i$'s endowment. We assume that $Y_{i}$ has a finite\protect\footnote{
All the results in this paper can be extended to the case that $Y_{i}$ has a compact support.} support $\mathcal{Y}_{i}$ and a probability mass function $e_{i}(y) : =\Pr (Y_{i} =y)$ for all $1 \leq i \leq n$ and $y \in \mathcal{Y}_{i}$. Let $\mathcal{Y} =\mathcal{Y}_{1} \times  . . . \times \mathcal{Y}_{n}$ and let $e(\boldsymbol{y})=e(y_{1},\ldots,y_{n}) : =\Pr (Y_{1} =y_{1} ,\ldots,Y_{n} =y_{n})$ be the joint probability mass function where we denote elements in $\mathbb{R}^{n}$ by bold letters. In each period $t =1 ,2 ,3\ldots $ the agents receive an endowment vector $\boldsymbol{y} \in \mathcal{Y}$ with probability $e(\boldsymbol{y}) > 0$. We assume that $\boldsymbol{y} \gg 0$  for all $\boldsymbol{y} \in \mathcal{Y}$ where $\boldsymbol{y} \gg 0$ means that $y_{i} >0$ for all $i =1 ,\ldots,n$, and that there exists an element $\boldsymbol{y}'$ in $\mathcal{Y}$ such that $\boldsymbol{y}' - \boldsymbol{y} \gg 0$ for all $\boldsymbol{y} \in \mathcal{Y} \setminus \{ \boldsymbol{y}' \}$. We refer to $e$ as the endowments process. 

Denote the agents' wealth at time $t =1$ by $a (1)$. In each period $t =1 ,2 ,3\ldots $, after receiving their endowment vector $\boldsymbol{y}(t)$, the agents choose a bundle of goods to consume in that period ($\boldsymbol{x}(t) =(x_{1}(t) ,\ldots,x_{n}(t)) \in \mathbb{R}_{ +}^{n}$)\protect\footnote{
As usual,  the positive cone of $\mathbb{R}^{n}$ is denoted by $\mathbb{R}_{ +}^{n}$, i.e., $\mathbb{R}_{ +}^{n} =\{\boldsymbol{x} =(x_{1} ,\ldots  ,x_{n}) :x_{j} \geq 0$ holds for all $j =1 ,\ldots  ,n\}$. 
} and choose how much to save in a risk-free bond for future consumption. The price of good $x_{i}(t)$ is given by $p_{i}(t) >0$, so the price of a bundle $\boldsymbol{x}(t)$ is $\boldsymbol{p}(t) \cdot \boldsymbol{x}(t)$ where $\boldsymbol{p}(t) \cdot \boldsymbol{x}(t) : =\sum p_{i}(t)x_{i}(t)$  denotes the scalar product of two elements in $\mathbb{R}^{n}$. The agents' savings rate of return is $1 +r(t)$ where $r(t)$ is the interest rate in period $t$. The agents are price takers, i.e., they take the sequence of prices $\{\boldsymbol{p}(t) ,r(t)\}_{t=1}^{\infty}$ as given where $\boldsymbol{p}(t) \gg 0$ and $r(t)>0$ for all $t$. If an agent's wealth at time $t$ is $a (t)$, the agent's wealth at time $t +1$ when $\boldsymbol{y}(t +1)$ is the realized endowment vector is 
\begin{equation*}a (t +1) =(1 +r(t)) (a (t) -\boldsymbol{p}(t) \cdot \boldsymbol{x}(t)) + \boldsymbol{p}(t +1) \cdot \boldsymbol{y} (t +1) . 
\end{equation*}
We assume that the agents can borrow, and the borrowing limit is given by $\underline{b}(t)$. Thus, $a(t) -\boldsymbol{p}(t) \cdot \boldsymbol{x}(t) \geq \underline{b}(t)$ for each period $t$. We assume that the borrowing limit in period $t$ is given by $\underline{b}(t) = -\frac{ \min _{\boldsymbol{y} \in \mathcal{Y}}\boldsymbol{p}(t) \cdot \boldsymbol{y}}{r(t)}$. In the stationary environment that we will study in the next section, the borrowing limit $\underline{b}(t)$ equals the natural borrowing limit (see \cite{aiyagari1994}).  


We denote by
$C(a,\boldsymbol{p}(t)) =[\underline{b}(t)\min \{a ,\sum _{i =1}^{n}p_{i}\left (t\right )\overline{b} /(1-r)^{2} \}]$ the interval from which an agent may choose his level of savings when his wealth is $a$ and the prices of goods are $\boldsymbol{p}(t)$.  $\sum _{i =1}^{n}p_{i}(t)\overline{b} / (1-r)^{2}$ is an upper bound on savings that ensure compactness of the state space where $\overline{b}>0$. Because we study a ``real economy" the upper and lower bounds on savings depend on the prices. We assume that the maximal level of savings that an agent can have is bounded to avoid technical difficulties that arise in dynamic programming with unbounded rewards and to provide an upper bound on savings that ensure existence when using numerical methods to find the equilibrium (see Remark \ref{Remark: compactness} for a discussion on  the upper bound on savings). Using known results from the previous literature on the income fluctuation problem, the upper bound on savings can be chosen so that it never binds (see Remark \ref{Remark: compactness}).

We assume that the agents' preferences over bundles are represented by a utility function $U:\mathbb{R}_{ +}^{n} \rightarrow \mathbb{R}$. For $\boldsymbol{x} ,\boldsymbol{x}^{ \prime } \in \mathbb{R}^{n}$ we write $\boldsymbol{x} \geq \boldsymbol{x}^{ \prime }$ if $x_{i} \geq x_{i}^{ \prime }$ for all $i =1 , . . . ,n$. We say that $U$ is increasing if $\boldsymbol{x} \geq \boldsymbol{x}^{ \prime }$ implies $U(\boldsymbol{x}) \geq U(\boldsymbol{x}^{ \prime })$. We say that $U$ is strictly increasing if $\boldsymbol{x} > \boldsymbol{x}^{ \prime }$ implies $U(\boldsymbol{x}) > U(\boldsymbol{x}^{ \prime })$.
Throughout the paper, we assume the following standard conditions on the utility function.

\begin{assumption}
\label{Assumption 0} (i) The utility function $U$ is strictly increasing, continuously differentiable, strictly concave, and $\frac{ \partial U(0)}{ \partial x_{j}} =\infty $ for some $1 \leq j \leq n$. 

(ii) For every $\boldsymbol{p} \gg 0$, the indirect utility function $u_{\boldsymbol{p}}(c) := U(\boldsymbol{x}^{ \ast }(c ,\boldsymbol{p}))$ is continuously differentiable and satisfies the following Inada conditions:  $\lim _{c \rightarrow \infty } u_{\boldsymbol{p}}' (c) =0$ and $\lim _{c \rightarrow 0 } u_{\boldsymbol{p}}' (c) =\infty$.\footnote{The assumption that $u$ is continuously differentiable can be relaxed (part (i) implies that  $u_{\boldsymbol{p}}(c)$ is concave so we can replace the derivative by the right-hand-derivative). }

\end{assumption}

Let $A$ be the set of possible wealth levels that an agent can have, and let $A^{t}:=\underbrace{A \times \ldots  \times A}_{t~ \mathrm{t} \mathrm{i} \mathrm{m} \mathrm{e} \mathrm{s}}$. A strategy $\pi $ for the agents is a function that assigns to every finite history $\boldsymbol{a}^{t} =(a(1) , . . . ,a(t)) \in A^{t}$ a feasible bundle $\boldsymbol{x}(t)$. A strategy $\pi $ induces a probability measure over the space of all infinite histories.\protect\footnote{
The probability measure on the space of all infinite histories $A^{\mathbb{N}}$ is uniquely defined (see for example \cite{bertsekas1978stochastic}).
} We denote the expectation with respect to that probability measure by $\mathbb{E}_{\pi }$.

When the agents follow a strategy $\pi $ and the sequence of prices is given by $\{\boldsymbol{p}(t) ,r(t)\}_{t=1}^{\infty}$, their expected present discounted value is
\begin{equation*}V_{\pi } (a) =\mathbb{E}_{\pi } \Big(\sum \limits _{t =1}^{\infty} \beta ^{t-1}U(\pi (a\left (1\right ),\ldots,a\left (t\right ))\Big),
\end{equation*}  where $a\left (1\right ) =a$ is the initial wealth and $1/2 <\beta  <1$ is the agents' discount factor. Denote
\begin{equation*}V (a) =\sup_{\pi }V_{\pi } (a).
\end{equation*}
That is, $V (a)$ is the maximal expected utility that an agent can have when his initial wealth is $a$. We call $V$ the value function.

\subsection{Competitive stationary equilibrium}
In this section we define a competitive stationary equilibrium (CSE). We first introduce some notations that are necessary in order to define a CSE. In a CSE the prices of goods and the interest rate are constant over time. For the rest of the section we assume that $(\boldsymbol{p}(t) ,r(t)) =(\boldsymbol{p} ,r)$ for all $t \in \mathbb{N}$.  

We denote by $b$ the agents' savings in the next period. When the agents' wealth is $a$, their next period's savings are $b \in C(a ,\boldsymbol{p})$, and the prices of goods are $\boldsymbol{p}$, then the set of consumption bundles available to the agents is given by $X(a-b ,\boldsymbol{p}) =\{x \in \mathbb{R}_{+}^{n} :\boldsymbol{p} \cdot \boldsymbol{x} =a -b\}$. We sometimes change variables and define $c=a-b$ to the total consumption of the agents. 

The minimal level of wealth that an agent can have when $\boldsymbol{p} \gg 0$ and $r \in (0,1)$ is 
$$\underline{a}(\boldsymbol{p},r)  = (1 +r)\underline{b} + \min_{\boldsymbol{y} \in \mathcal{Y}}\boldsymbol{p} \cdot \boldsymbol{y} = -\frac{(1 +r) \min _{\boldsymbol{y} \in \mathcal{Y}} \boldsymbol{p} \cdot \boldsymbol{y} } {r} + \min_{\boldsymbol{y} \in \mathcal{Y}}\boldsymbol{p} \cdot \boldsymbol{y} = - \frac { \min_{\boldsymbol{y} \in \mathcal{Y}}\boldsymbol{p} \cdot \boldsymbol{y} } {r} $$ 
and the maximal level of wealth that an agent can have is $$\overline{a}(\boldsymbol{p},r) : = (1 +r)\sum _{i =1}^{n}\frac{ p_{i}\overline{b} } { (1-r)^{2}} +\max _{\boldsymbol{y} \in \mathcal{Y}}\boldsymbol{p} \cdot \boldsymbol{y}.$$ Hence, the set of possible wealth levels that an agent can have $$A(\boldsymbol{p},r) :=[\underline{a}(\boldsymbol{p},r) ,\overline{a}(\boldsymbol{p},r)]$$ 
is compact for all $\boldsymbol{p} \gg 0$ and $r \in (0,1)$. For the rest of the section we assume that  $\boldsymbol{p} \gg 0$ and $r \in (0,1)$.

Because for any $\boldsymbol{p} \gg 0$ and $r \in (0,1)$ the value function is bounded on the compact set $A(\boldsymbol{p},r)$, we can use standard dynamic programming arguments to solve the agents' problem. Let $B(A)$ be the space of all bounded real-valued functions defined on a set $A$. For any  $\boldsymbol{p} \gg 0$ and $r \in (0,1)$, define the operator $T:B(A(\boldsymbol{p},r)) \rightarrow B(A(\boldsymbol{p},r))$ by \begin{equation*}Tf(a,\boldsymbol{p},r) =\max _{b \in C(a,\boldsymbol{p})}\max _{\boldsymbol{x} \in X(a -b ,\boldsymbol{p})}U(\boldsymbol{x}) +\beta \sum \limits _{\boldsymbol{y} \in \mathcal{Y}}^{\,}e(\boldsymbol{y})f((1 +r)b +\boldsymbol{p} \cdot \boldsymbol{y} ,\boldsymbol{p} ,r).
\end{equation*}
The value function $V$ is the unique fixed point of $T$, i.e., there is a unique function $V \in B \left (A(\boldsymbol{p},r)\right )$ such that $T V =V$.\protect\footnote{
The Banach-fixed point theorem (see Theorem 3.48 in \cite{aliprantis2006infinite}) shows that $T$ has a unique fixed point. Standard dynamic programming arguments (e.g., \cite{blackwell1965}) show that the value function $V$ is the unique fixed point of $T$.} 

We denote by $\boldsymbol{x}^{ \ast }(a -b ,\boldsymbol{p})$ the demand function of an agent, i.e., \begin{equation*}\boldsymbol{x}^{ \ast }(a -b ,\boldsymbol{p}) =\ensuremath{\operatorname*{argmax}}_{\boldsymbol{x} \in X(a -b ,\boldsymbol{p})}U(\boldsymbol{x}).\label{eq:demand function}\end{equation*}
Note that given the choice of the next's period savings $b$, the decision of how to distribute the spending $a-b$ between the different consumption goods is a static decision. 
Also note that $\boldsymbol{x}^{ \ast }$ is single-valued since $U$ is strictly concave and continuous. We denote by $g(a ,\boldsymbol{p} ,r)$ the savings policy function, i.e., 
\begin{equation}g(a ,\boldsymbol{p} ,r) =\ensuremath{\operatorname*{argmax}}_{b \in C(a ,\boldsymbol{p})} \ U(\boldsymbol{x}^{ \ast }(a -b ,\boldsymbol{p})) +\beta \sum \limits _{\boldsymbol{y} \in \mathcal{Y}}^{\,}e(\boldsymbol{y})V((1 +r)b +\boldsymbol{p} \cdot \boldsymbol{y} ,\boldsymbol{p} ,r) . \label{eq:savings policy}
\end{equation}
A standard dynamic programming argument (see Lemma \ref{Lemma 1} in the Appendix) shows that the savings policy function is continuous and single-valued under Assumption \ref{Assumption 0}.

\begin{remark} \label{Remark: compactness}
For simplicity, we assume an exogenous upper bound on savings to prove that the value function satisfies the Bellman equation and the savings policy function exists. Given previous known results derived in the standard income fluctuation problem, this assumption is not, however, necessary in order to prove that the value function satisfies the Bellman equation. Because the decision of how to distribute consumption between the different goods is static, for a fixed vector of positive prices $\boldsymbol{p}$, we can reduce the single agent dynamic optimization problem to a standard income fluctuation problem by defining the indirect utility function  $u_{\boldsymbol{p}}(c) := U(\boldsymbol{x}^{ \ast }(c ,\boldsymbol{p}))$. In this case, we can use the Coleman operator approach to show that under mild conditions, the value function and the savings policy function exist and the upper bound on savings can be chosen so that it never binds (see \cite{li2014solving}, \cite{accikgoz2018existence}, \cite{ma2020income} and references therein). Another reason that we assume an upper bound on savings is that when using numerical methods to find the equilibrium, we must have an upper bound on savings. In this case, we usually can't choose ex-ante an upper bound that does not bind. In contrast to Bewley models with one consumption good, in our model the choice of an upper bound and a lower bound is important for guaranteeing the  existence of a CSE because it influences the properties of the excess demand function such as homogeneity and the boundary conditions (see Remark \ref{Remark Borrowing}). 
\end{remark}

For a set $K \subseteq \mathbb{R}^{n}$ we denote by $\mathcal{P}(K)$ the set of all probability measures on $K$ and by $\mathcal{B}(K)$ the Borel sigma-algebra on $K$. Define
\begin{equation}M\lambda (D;\boldsymbol{p} ,r) =\int_{A(\boldsymbol{p},r)}\sum \limits _{\boldsymbol{y} \in \mathcal{Y}}^{\,}e(\boldsymbol{y})1_{D}((1 +r)g(a ,\boldsymbol{p} ,r) +\boldsymbol{p} \cdot \boldsymbol{y})\lambda (da;\boldsymbol{p} ,r) , \label{eq: operator M}\end{equation} 
for any $D\in\mathcal{B}(A(\boldsymbol{p},r))$ where $1_{D}$ is the indicator function of the set $D \in \mathcal{B}(A(\boldsymbol{p},r))$. $M\lambda  \in \mathcal{P}(A(\boldsymbol{p},r))$ describes the next period's wealth distribution, given that the current wealth distribution is $\lambda  \in \mathcal{P}(A(\boldsymbol{p},r))$ and the prices are $(\boldsymbol{p} ,r)$. A wealth distribution $\mu  \in \mathcal{P}(A(\boldsymbol{p},r))$ is called an invariant wealth distribution if $\mu=M\mu $.  

We now define a CSE.

\begin{definition} \label{Def: CSE}
A competitive stationary equilibrium consists of prices $(\boldsymbol{p} ,r)$, a savings policy function $g$, a demand function $\boldsymbol{x}^{ \ast }$, and a wealth distribution $\mu  \in \mathcal{P}(A(\boldsymbol{p},r))$ such that 

(i) Given the prices $(\boldsymbol{p} ,r)$, the savings policy function $g$ and the demand function $\boldsymbol{x}^{ \ast }$ are optimal for the agents. That is, $g$ satisfies equation (\ref{eq:savings policy}) and \begin{equation*}\boldsymbol{x}^{ \ast }(a -g(a ,\boldsymbol{p} ,r) ,\boldsymbol{p}) =\ensuremath{\operatorname*{argmax}}_{\boldsymbol{x} \in X(a -g(a ,\boldsymbol{p} ,r) ,\boldsymbol{p})}U(\boldsymbol{x}) .
\end{equation*}

(ii) Given the prices $(\boldsymbol{p} ,r)$, $\mu $ is an invariant wealth distribution. That is, $\mu  \in \mathcal{P}(A(\boldsymbol{p},r))$ satisfies $\mu  =M\mu $.

(iii) For each good $1 \leq i \leq n$, the aggregate supply of good $i$ equals the aggregate demand for good $i$: 
\begin{equation*}\int _{A(\boldsymbol{p},r)}x_{i}^{ \ast }(a -g(a ,\boldsymbol{p} ,r) ,\boldsymbol{p})\mu (da;\boldsymbol{p} ,r) =\sum \limits _{y_{i} \in \mathcal{Y}_{i}}^{\,}e_{i}(y_{i})y_{i}
\end{equation*}

(iv) The aggregate supply of savings equals the aggregate demand for savings:
\begin{equation*}\int _{A(\boldsymbol{p},r)}g(a ,\boldsymbol{p} ,r)\mu (da;\boldsymbol{p} ,r) =0.
\end{equation*}
\end{definition}

The first equilibrium condition says that agents choose a demand function and a savings policy function to maximize their expected discounted utility. The second equilibrium condition says that the wealth distribution induced by the agents' savings policy function is invariant. The third equilibrium condition says that the aggregate demand for good $i$ equals the aggregate supply of good $i$. The fourth equilibrium condition says that the aggregate savings in the economy are $0$, i.e., the supply of savings equals the demand for savings. The third and fourth equilibrium conditions require that the prices $(\boldsymbol{p},r)$ are market-clearing prices. The natural interpretation of the stationary equilibrium prices are that the prices represent average prices (see \cite{huggett1993risk}). An alternative to CSE is a competitive recursive equilibrium (see \cite{miao2006competitive}). A competitive recursive equilibrium is a sequence of prices $(\boldsymbol{p}(t),r(t))$, and a sequence of measures $(\lambda (t))$ such that the savings and consumption decisions are optimal for the agents; the prices $(\boldsymbol{p}(t) ,r(t))$ are market-clearing prices for every period $t$; and the wealth distribution follows the law of motion defined by equation (\ref{eq: operator M}). Clearly, if the initial agents' wealth distribution is invariant, then the CSE is also a competitive recursive equilibrium. In this paper we focus on a CSE. The existence result presented in the next section can be applied to the competitive recursive equilibrium case as well.\protect\footnote{
For general existence results of a competitive recursive equilibrium with aggregate shocks see \cite{brumm2017recursive}.} The analysis and computation of a competitive recursive equilibrium are generally much harder than the analysis and computation of a CSE.

We note that the model presented in this paper is closely related to Huggett's model \citep{huggett1993risk} and Bewley's model \citep{bewley1986stationary}. In Huggett's model there is only one consumption good and only the interest rate is determined in equilibrium. In Bewley's model there are many consumption goods and their prices are determined in equilibrium but the interest rate is fixed and is not determined in equilibrium. In the model presented in this paper, however, there are many consumption goods, and both their prices and the interest rate are determined in equilibrium. Our model also generalizes the static pure-exchange Arrow-Debreu model to an incomplete-market economy where the agents can transfer assets from one period to another only by investing in a risk-free bond.

\section{Main results}
In this section we present our main results. In Section 3.1 we state our existence result. In Section 3.2 we provide conditions that ensure the uniqueness of a competitive stationary equilibrium (CSE). In Section 3.3 we discuss how an increase in the risk of the endowments process influences the equilibrium prices of goods and the equilibrium interest rate when the agents' preferences can be represented by a CES utility function. In Section 3.4 we compare our model to mean field equilibrium models. In Section 3.5 we extend the model to include ex-ante heterogeneous agents.

\subsection{Existence of a CSE} \label{SEC: Existence}

The main theorem of this section is the following:
\begin{theorem}
\label{Theorem exist} Suppose that Assumption \ref{Assumption 0}  holds. Then, there exists a competitive stationary equilibrium. 
\end{theorem}

To prove the theorem, we construct a natural excess demand function and show that the function has a zero. Under Assumption \ref{Assumption 0}, the savings policy function $g(a ,\boldsymbol{p} ,r)$ is single-valued and continuous. Furthermore, there exists a unique invariant wealth distribution $\mu (da;\boldsymbol{p} ,r)$ for all $(\boldsymbol{p} ,r) \in \boldsymbol{P}$ where $\boldsymbol{P}$ is the non-empty and convex set defined in equation (\ref{eq: P}) below. We define an excess demand function $\zeta (\boldsymbol{p} ,r)$ from $\boldsymbol{P} \subseteq \mathbb{R}_{ +}^{n +1}$ into $\mathbb{R}^{n +1}$, where 
\begin{equation*}\zeta _{i}(\boldsymbol{p} ,r) =\int _{A(\boldsymbol{p},r)}x_{i}^{ \ast }(a -g(a ,\boldsymbol{p} ,r) ,\boldsymbol{p})\mu (da;\boldsymbol{p} ,r) -\sum \limits _{y_{i} \in \mathcal{Y}_{i}}^{\,}e_{i}(y_{i})y_{i}
\end{equation*} 
is the excess demand for good $i$, $i =1 , . . . ,n$, and 
\begin{equation*}\zeta _{n +1}(\boldsymbol{p} ,r) = -\int _{A(\boldsymbol{p},r)}g(a ,\boldsymbol{p} ,r)\mu (da;\boldsymbol{p} ,r)\end{equation*} 
is the excess demand for savings. The excess demand function $\zeta  :\boldsymbol{P} \rightarrow \mathbb{R}^{n +1}$ is defined by     \begin{equation*}\zeta (\boldsymbol{p} ,r) =(\zeta _{1}(\boldsymbol{p} ,r) , \ldots ,\zeta _{n}(\boldsymbol{p} ,r) ,\zeta _{n +1}(\boldsymbol{p} ,r)) .
\end{equation*}
Note that if $\zeta (\boldsymbol{p} ,r) =\boldsymbol{0}$ then $(\boldsymbol{p} ,r)$ are equilibrium prices, $\mu ( \cdot ;\boldsymbol{p} ,r)$ is the equilibrium invariant wealth distribution, $\boldsymbol{x}^{ \ast }(a -g(a ,\boldsymbol{p} ,r) ,\boldsymbol{p})$ is the equilibrium demand function, and $g(a ,\boldsymbol{p} ,r)$ is the equilibrium savings policy function. 

To prove that the excess demand function has a zero we prove properties of the excess demand function. We first extend a well-known result from the static Arrow-Debreu model to the incomplete-markets economy studied in this paper. We show that the aggregate demand for goods and the aggregate supply of goods depend only on their relative prices. In particular, the next Proposition shows that if $(\boldsymbol{p} ,r)$ are equilibrium prices then $(\theta \boldsymbol{p} ,r)$ are also equilibrium prices for all $\theta  >0$. This result is trivial in the standard static Arrow-Debreu model, however, in the incomplete-market  Arrow-Debreu economy this result is challenging to establish because the excess demand function depends on the invariant wealth distribution. The key part of proving part (ii) of Proposition \ref{Prop homogenous of }  is to show that the invariant wealth distribution is homogeneous in the prices of goods for a suitable definition of homogeneity that we introduce for probability measures.

\begin{proposition}
\label{Prop homogenous of }Fix $\boldsymbol{p} \gg 0$, $ 0 < r <1$, and $\theta  >0$. Then 

(i) $\theta g(a ,\boldsymbol{p} ,r) =g(\theta a ,\theta \boldsymbol{p} ,r)$ and $\boldsymbol{x}^{ \ast }(a -g(a ,\boldsymbol{p} ,r) ,\boldsymbol{p}) =\boldsymbol{x}^{ \ast }(\theta a -g(\theta a ,\theta \boldsymbol{p} ,r) ,\theta \boldsymbol{p})$ for all $a$.

(ii) $\zeta _{i}(\theta \boldsymbol{p} ,r) =\zeta _{i}(\boldsymbol{p} ,r)$ for $1 \leq i \leq n$ and $\zeta _{n +1}(\theta \boldsymbol{p} ,r) =\theta \zeta _{n +1}(\boldsymbol{p} ,r)$. 

Thus, if $(\boldsymbol{p} ,r)$ are equilibrium prices then $(\theta \boldsymbol{p} ,r)$ are also equilibrium prices. 
\end{proposition}

We note that the excess demand for savings is homogeneous of degree one in the prices of goods, while the excess demand for good $k$ is homogeneous of degree zero for all $1 \leq k \leq n$. These results use  the fact that 
$\theta C(a,\boldsymbol{p}) = C(\theta a,\theta \boldsymbol{p})$ for all $\theta > 0$ where \begin{equation*}
C(a,\boldsymbol{p}) =  \left [- \frac{ \min _{\boldsymbol{y} \in \mathcal{Y}}\boldsymbol{p} \cdot \boldsymbol{y}}{r},\min \left \{a ,\sum _{i =1}^{n}\frac{ p_{i}\overline{b}} {(1-r)^{2} } \right \} \right ]
\end{equation*}
is the interval from which an agent may choose his level of savings and $\theta A =\{\theta x :x \in A\}$ for any set $A$. That is, if the agent can save an amount  $b$ given the wealth level $a$ and the prices $\boldsymbol{p}$, then the agent can save an amount $\theta b$ given the wealth level $\theta a$ and the prices $\theta \boldsymbol{p}$. This is reasonable in our setting since all the prices in our model are real prices (see also Remark \ref{Remark: compactness}).    
Also note that if the aggregate savings do not equal zero in equilibrium, then it is not true that if $(\boldsymbol{p},r)$ are equilibrium prices then $(\theta\boldsymbol{p},r)$ are also equilibrium prices for all $\theta >0$. This follows because the aggregate demand for savings is homogeneous of degree one in the prices of goods. Hence, in an economy with production Proposition \ref{Prop homogenous of } can fail (unless the demand for capital is also homogeneous).  

From Proposition \ref{Prop homogenous of }, if $(\boldsymbol{p},r)$ are equilibrium prices then $(\theta\boldsymbol{p},r)$ are also equilibrium prices for all $\theta >0$. Thus, we can normalize the prices of the goods. More precisely, the search for equilibrium prices can be confined to sets that contain at least one element from the half-ray $\{\theta \boldsymbol{p} :\theta  >0\}$.

We define the sets $\Lambda  =\{(\boldsymbol{p} ,r) \in \mathbb{R}_{ +}^{n} \times \mathbb{R}_{ +} :\sum _{i =1}^{n}p_{i} +r =1\}$ and 
\begin{equation} \boldsymbol{P} =\{(\boldsymbol{p} ,r) \in \Lambda  :\boldsymbol{p} \gg 0 ,r >0\} . \label{eq: P}
\end{equation} 
In order to prove the existence of a CSE, that is, to prove that there are prices $(\boldsymbol{p} ,r)\in \boldsymbol{P}$ such that  $\zeta (\boldsymbol{p} ,r) =\boldsymbol{0}$, we show that the excess demand function is continuous, satisfies Walras' law and satisfies suitable boundness and boundary conditions (see Proposition \ref{Prop: existence} in the Appendix). We then show that these conditions guarantee the existence of at least one vector $(\boldsymbol{p} ,r) \in \boldsymbol{P}$ that satisfies $\zeta (\boldsymbol{p} ,r) =\boldsymbol{0}$.   

\begin{remark}\label{Remark Borrowing} We prove the existence of an equilibrium with a strictly positive interest rate. An equilibrium with a strictly positive interest rate exists since we assume that the borrowing constraint tends to minus infinity as the interest rate tends to zero (a similar observation is made on page 673 in \cite{aiyagari1994}). Also note that for $(\boldsymbol{p},r) \in \boldsymbol{P}$ we have $\sum _{i=1}^{n} p_{i} / (1-r)^{2} = 1/\sum _{i=1}^{n} p_{i}$ so the upper bound on savings tends to infinity when $\sum _{i=1}^{n} p_{i}$ tends to $0$. These conditions are important for proving that the excess demand function satisfies appropriate boundary conditions. 
\end{remark}

\subsection{Uniqueness of a CSE}
In this section we prove the uniqueness of a CSE for the special case where the utility function is given by: $U(\boldsymbol{x}) =\sum _{i =1}^{n}\alpha _{i}x_{i}^{\gamma }$ for some $0 <\gamma  <1$, $\alpha _{i} >0$, $\sum _{i =1}^{n}\alpha _{i} =1$, i.e., the agents' preferences over bundles can be represented by a CES utility function with an elasticity of substitution that is higher than one. 

There is a vast literature that provides sufficient conditions to ensure the uniqueness of an equilibrium in the standard static pure-exchange Arrow-Debreu model.\protect\footnote{
For a survey of the work done on the uniqueness of equilibrium, see  \cite{arrow1971competitive}, \cite{mas1991uniqueness}, and \cite{kehoe1998uniqueness}. For recent results, see \cite{toda2017edgeworth} and \cite{geanakoplos2018uniqueness}, and references therein.}  The property that the demand for each good increases with the prices of the other goods (``gross substitutes property'') usually plays a crucial role in proving the uniqueness of an equilibrium in the static Arrow-Debreu model. Given the gross substitutes property, an easy argument shows that the equilibrium must be unique. This fact led most of the previous literature on the uniqueness of an equilibrium to find conditions on agents' preferences that ensure that the gross substitutes property holds. While the gross substitutes property remains an important property in proving the uniqueness of an equilibrium in the dynamic incomplete-market Arrow-Debreu model considered in this paper also, the standard argument that proves the uniqueness of an equilibrium does not apply. The reason is that the aggregate demand for goods does not necessarily increase with the interest rate, and the aggregate demand for savings does not necessarily increase with the prices of goods. Thus, the excess demand function does not necessarily have the gross substitutes property. The coupling of the aggregate savings and the aggregate demand for goods leads to a complicated behavior of the excess demand function. Nonetheless, we show that when the agents' preferences are represented by a CES utility function with an elasticity of substitution that is higher than one the CSE is unique even when the excess demand does not have the gross substitutes property.

It is well known and easy to check that when the agents' preferences are represented by a CES utility function, the indirect utility function \begin{equation*}
 v(a -b ,\boldsymbol{p}) =\max _{x \in X(a -b ,\boldsymbol{p})}U(\boldsymbol{x})
\end{equation*} 
is given by a constant relative risk aversion (CRRA) utility function. The CRRA utility function is popular in the applied literature and is often used in numerical analysis of incomplete markets heterogeneous agent models. The uniqueness of an equilibrium in these models with one consumption good and a CRRA utility function has recently been studied in \cite{proehl2018existence} and \cite{light2017uniqueness}. The next theorem generalizes the results in \cite{light2017uniqueness} to a model with many consumption goods.    

\begin{theorem}
\label{Theorem uniq}Assume that $U(\boldsymbol{x}) =\sum _{i =1}^{n}\alpha _{i}x_{i}^{\gamma }$ for some $0 <\gamma  <1$, $\alpha _{i} >0$, $\sum _{i =1}^{n}\alpha _{i} =1$. Then there exists a unique competitive stationary equilibrium. 
\end{theorem}

If the agents' preferences can be represented by a  Cobb-Douglas utility function, i.e., $U(\boldsymbol{x}) =\sum _{i =1}^{n}\alpha _{i}\ln (x_{i})$ for $\alpha _{i} >0$, $\sum _{i =1}^{n}\alpha _{i} =1$, then the same proof as the proof of Theorem \ref{Theorem uniq} shows that there exists a unique CSE in this case, as well. Note that in this case, the indirect utility function corresponds to a log utility function, which is often used in the quantitative literature (for example, see \cite{aiyagari1994} and \cite{krusell2010labour}).

The conditions on the agents' preferences that ensure uniqueness are restrictive.\protect\footnote{
A natural question that arises is: under what conditions is there a finite number of equilibria? (see \cite{debreu1970economies} for an answer to this question in the static Arrow-Debreu model). A related question can be asked about the stability of the CSE (see \cite{arrow1958stability} and \cite{arrow1959stability}). We did not explore these directions in the current paper.
} However, uniqueness results in Bewley models are rare, and a multiplicity of equilibria can arise even under the standard specifications of the model (for examples of the multiplicity of equilibria see \cite{toda2017huggett} and \cite{accikgoz2018existence}). Even in a static Arrow-Debreu model, a multiplicity of equilibria can easily arise.
\cite{kubler2010competitive} and \cite{kubler2010tackling} provide examples of multiplicity in the case that the agents' preferences can be represented by a CES utility function, and also provide a general method of finding all the equilibria in semi-algebraic Arrow-Debreu models.

\subsection{The price of goods and wealth inequality}

In this section we show that if the agents' preferences are represented by a CES  utility function with an elasticity of substitution that is higher than one, then an increase in the risk of the endowments process (in the sense of the convex stochastic order) does not change the equilibrium prices of goods, and decreases the equilibrium interest rate. 

In response to an increase in the risk of the endowments process, we show that the partial equilibrium wealth inequality is higher in the sense of the convex stochastic order. That is, for a fixed interest rate and fixed prices of goods, the wealth inequality is higher when the endowments process is riskier. This follows from the facts that agents save more when the future endowment is riskier and that the savings policy function is convex in wealth, i.e., the marginal propensity to consume is decreasing.\footnote{The convexity of the savings policy function follows from the CES assumption and is not easy to establish for general utility functions.} In addition, for a fixed interest rate $r$ and prices of goods $\boldsymbol{p}$, the precautionary savings effect increases the aggregate savings, and thus the aggregate expenditure on goods decreases. Since the goods are normal, the decrease in the aggregate expenditure on goods implies that the aggregate demand for each good decreases. 

We note that this is different from the static Arrow-Debreu model where riskier endowments do not change the demand for each good. In the static Arrow-Debreu model the demand for each good is linear in wealth while in our setting the demand for each good is concave in wealth because the marginal propensity to consume is decreasing. Thus, in the dynamic incomplete-market Arrow-Debreu model the equilibrium prices might change in response to an increase in the risk of the endowments process. The prices of goods, however, do not change at all in the new CSE. While the interest rate decreases, the negative effect of this decrease on the aggregate savings is exactly offset by the positive effect on the aggregate savings of an increase in the risk of the endowments process. In other words, the negative substitution effect on the aggregate savings and the positive precautionary effect on the aggregate savings are equal. 


We now introduce notations that are needed to state the main theorem of this section. For two probability measures $\lambda _{1} ,\lambda _{2}$ we define the partial order $ \succeq _{I -CX}$ by $\lambda _{2} \succeq _{I -CX}\lambda _{1}$ if and only if $\int f(a)\lambda _{2}(da) \geq \int f(a)\lambda _{1}(da)$ for every convex and increasing function $f$. Similarly, we write $\lambda _{2} \succeq _{CX}\lambda _{1}$ if and only if $\int f(a)\lambda _{2}(da) \geq \int f(a)\lambda _{1}(da)$ for every convex function $f$. We say that the endowments process $e$ is riskier than the endowments process $e^{ \prime }$ if $e \succeq _{CX}e^{ \prime }$. With slight abuse of notation, we add the argument $e$ to the functions defined above, when $e(\boldsymbol{y})$ is the probability of receiving the endowment vector $y \in \mathcal{Y}$. For example, we write $\mu ( \cdot ;\boldsymbol{p} ,r ,e)$ for the invariant wealth distribution, $g(a ,\boldsymbol{p} ,r ,e)$ for the savings policy function, and $\boldsymbol{x}^{ \ast }(a -g(a ,\boldsymbol{p} ,r ,e) ,\boldsymbol{p})$ for the demand function.

\begin{theorem}
\label{Theorem wealth ineq}Assume that $U(\boldsymbol{x}) =\sum _{i =1}^{n}\alpha _{i}x_{i}^{\gamma }$ for some $0 <\gamma  <1$, $\alpha _{i} >0$, $\sum _{i =1}^{n}\alpha _{i} =1$. Assume that the endowments process $e$ is riskier than the endowments process $e^{ \prime }$. Then

(i) The partial equilibrium wealth inequality is higher under $e$ than under $e^{ \prime }$, i.e., $\mu ( \cdot ;\boldsymbol{p},r,e) \succeq _{I-CX}\mu ( \cdot ;\boldsymbol{p} ,r ,e^{ \prime })$ for all $(\boldsymbol{p} ,r) \in \boldsymbol{P}$. In addition, if  $(\boldsymbol{p}(e),r(e))$ are equilibrium prices under the endowments process $e$ then $\mu ( \cdot ;\boldsymbol{p}(e),r(e),e) \succeq _{CX}\mu ( \cdot ;\boldsymbol{p}(e) ,r(e) ,e^{ \prime })$.  

(ii) The equilibrium prices of goods do not change, i.e., $\boldsymbol{p}(e) =\boldsymbol{p}(e^{ \prime })$. The equilibrium interest rate is lower under $e$ than under $e^{ \prime }$, i.e., $r(e^{ \prime }) \geq r(e)$.
\end{theorem}

We note that when the endowments process $e$ is riskier than the endowments process $e^{\prime}$ then the total supply of each consumption good does not change and the relative total supply of each consumption good does not change either (see more details in the proof of Theorem \ref{Theorem wealth ineq}). This fact plays a major rule in the proof of Theorem \ref{Theorem wealth ineq}, in particular, in proving that the prices of consumption goods do not change.

When the agents' preferences are not represented by a CES utility function, Theorem \ref{Theorem wealth ineq} does not necessarily hold. In future research, it would be interesting to explore the connection between the prices of the consumption goods and the risk of the endowment process for different utility functions.

\subsection{Comparison to mean field equilibrium models} \label{sec:MFE}
Mean field equilibrium models have been popularized in the recent literature in operations research, economics, and optimal control (e.g., see \cite{lasry2007mean}, \cite{weintraub2008markov}). 
In a mean field model, the agents' utility functions and the evolution of the agents' states depend on the distribution of the other agents' states. In a mean field equilibrium, each agent maximizes his expected discounted payoff, assuming that the distribution of the other agents' states is fixed. Given the agents' strategy, the distribution of the agents' states is an invariant distribution of the Markov process that governs the dynamics of the agents' states. While the notion of a mean field equilibrium is conceptually similar to the notion of a CSE, we cannot write the dynamic incomplete markets model studied in this paper as a discrete-time mean field model. This is because the market-clearing conditions (see conditions (iii) and (iv) in Definition \ref{Def: CSE}) are not consistent with the definition of a mean field equilibrium as the prices and the interest rate cannot be written as a function of the agents' states distribution (the wealth distribution in our setting).\footnote{Note that other heterogeneous agent macro models can be written as a discrete-time mean field model (see \cite{acemoglu2012} and \cite{LW2018}). In particular, the model in Aiyagari (1994) can be written as a discrete-time mean field model. } Thus, we cannot apply the recent existence, uniqueness and comparative statics results developed for discrete-time mean field equilibrium models (e.g., \cite{acemoglu2012}, and \cite{LW2018}).

\subsection{Ex-ante heterogeneous agents}

In this section we extend the model described in Section 2 to the case of ex-ante heterogeneous agents. We assume that the agents are heterogeneous
in their preferences over consumption bundles as well as in their endowments. Assume that before the process starts, each agent has a type $\theta  \in \Theta $. For simplicity we assume that $\Theta $ is a finite set. Each agent's type is fixed throughout the horizon. An agent with type $\theta  \in \Theta $ has preferences that are represented by a utility function $U (\boldsymbol{x} ,\theta )$ and receives an endowment $\boldsymbol{y} (\theta )$ with probability $e (\boldsymbol{y} (\theta ))$ in each period. Let $\phi $ be the probability mass function over the type space; $\phi  (\theta )$ is the mass of agents whose types are $\theta  \in \Theta $. Adding the argument $\theta  \in \Theta $ to the functions defined in Section 2, we can modify the definitions of Section 2 to include the ex-ante heterogeneity of agents. For example, $g (a ,\boldsymbol{p} ,r ,\theta )$ is the savings policy function of type $\theta $ agents and $\boldsymbol{x}^{ \ast } (a -g (a ,\boldsymbol{p} ,r ,\theta ) ,\boldsymbol{p} ,\theta )$ is the demand function of type $\theta $ agents.

Let $A_{h} =\mathbb{R} \times \Theta $ be an extended state space for the model with ex-ante heterogeneous agents. If an agent's extended state is $a_{h} =(a ,\theta ) \in A_{h}$ then the agent's wealth is $a$ and his type is $\theta $. Let $\lambda _{h}$ be a probability measure over the extended state space, i.e., $\lambda _{h} \in \mathcal{P}(A_{h})$. 

Define the Markov kernel
\begin{equation*}Q _{h}((a ,\theta ) ,D \times E) =\sum _{\boldsymbol{y} \in \mathcal{Y}}e (\boldsymbol{y}(\theta )) 1_{D} ((1 +r) g (a ,\boldsymbol{p} ,r ,\theta ) +\boldsymbol{p} \cdot \boldsymbol{y}(\theta ))1_{E}\text{}\left (\theta \right )
\end{equation*}for any $D \times E \in \mathcal{B}(\mathbb{R}) \times 2^{\Theta }$. The Markov kernel $Q_{h}$ describes the evolution of the extended state. That is, when the agent's wealth is $a$ and his type is $\theta $, the probability that the next period's pair of wealth-type will lie in $D \times E \in \mathcal{B}(\mathbb{R}) \times 2^{\Theta }$ is given by $Q _{h}((a ,\theta ) ,D \times E)$.

Define
\begin{equation*}M \lambda _{h} (D \times E ;\boldsymbol{p} ,r) =\int \sum _{\boldsymbol{y} \in \mathcal{Y}}^{\,}e (\boldsymbol{y}(\theta )) 1_{D} ((1 +r) g (a ,\boldsymbol{p} ,r ,\theta ) +\boldsymbol{p} \cdot \boldsymbol{y}(\theta )) 1_{E}\left (\theta \right )\lambda _{h} (d( a ,\theta ) ;\boldsymbol{p} ,r) ,
\end{equation*}for any $D \times E \in \mathcal{B}(\mathbb{R}) \times 2^{\Theta }$. $M \lambda _{h} \in \mathcal{P} (A_{h})$ describes the next period's wealth-types distribution, given that the current wealth-types distribution is $\lambda _{h} \in \mathcal{P} (A_{h})$ and the prices are $(\boldsymbol{p} ,r)$. A wealth-types distribution $\mu _{h} \in \mathcal{P} (A_{h})$ is called an invariant wealth-types distribution if $\mu _{h} =M \mu _{h}$.

These definitions map the model with ex-ante heterogeneous agents to the model with ex-ante homogeneous agents that we considered in Section 2. We can define a competitive stationary equilibrium as in Definition \ref{Def: CSE}. A competitive stationary equilibrium consists of prices $(\boldsymbol{p} ,r)$, savings policy functions $g$, demand functions $\boldsymbol{x}^{ \ast }$, and a wealth-types distribution $\mu _{h} \in \mathcal{P} (A_{h})$ such that the savings policy function $g$ and the demand function $\boldsymbol{x}^{ \ast }$ are optimal for each type $\theta $, the wealth-types distribution $\mu _{h}$ is invariant, and the prices of goods and the interest rate are market-clearing (see conditions (iii) and (iv) in Definition \ref{Def: CSE}).  

We note that if $U (\boldsymbol{x} ,\theta )$ satisfies Assumption \ref{Assumption 0} for each $\theta $, then Theorem \ref{Theorem exist} holds and there exists a CSE. The proof is similar to the proof of Theorem \ref{Theorem exist} so we omit the details.

\section{Final remarks}
In this paper we study a dynamic incomplete-market Arrow-Debreu economy which combines a Huggett-Bewley economy with the classic static pure-exchange Arrow-Debreu economy. We study a competitive stationary equilibrium where the prices of consumption goods and the interest rate are market clearing and the wealth distribution is invariant. Under mild conditions on the agents' preferences, we prove that the aggregate demand for consumption goods is homogeneous of degree 0, while the aggregate demand for savings is homogeneous of degree 1 (see Proposition \ref{Prop homogenous of }), and we prove the existence of a competitive stationary equilibrium (CSE) (see Theorem \ref{Theorem exist}). Under a CES utility function, we discuss how a riskier endowments process affects wealth inequality, the prices of goods and the interest rate. We prove that if the agents' preferences can be represented by a CES utility function with an elasticity of substitution that is equal to or higher than one, then there exists a unique CSE (see Theorem \ref{Theorem uniq}), and that a riskier endowments process increases the partial equilibrium wealth inequality, decreases the equilibrium interest rate, and does not change the equilibrium prices of goods (see Theorem \ref{Theorem wealth ineq}). It remains an open question whether Theorem \ref{Theorem uniq} and Theorem \ref{Theorem wealth ineq} can be extended to different utility functions. Many other open questions remain concerning the CSE. For example, studying the stability of a CSE awaits future research. 

\appendix 

\section{Appendix}

\subsection{Homogeneity of the excess demand function}

In this section we prove Proposition \ref{Prop homogenous of }.

Recall that for a set $K$ we denote by $\mathcal{P}(K)$ the set of all probability measures defined on $K$. We endow $\mathcal{P}(\mathbb{R})$ with the topology of weak convergence. We
say that $\lambda _{n} \in \mathcal{P} (\mathbb{R})$ converges weakly to $\lambda  \in \mathcal{P} (\mathbb{R})$ if for all bounded and continuous functions $f :\mathbb{R} \rightarrow \mathbb{R}$ we have
\begin{equation*}\underset{n \rightarrow \infty }{\lim }\int _{\mathbb{R}}f (a) \lambda _{n} (d a) =\int _{\mathbb{R}}f (a) \lambda  (d a)\text{.}
\end{equation*}

\noindent \textbf{Proposition \ref{Prop homogenous of }.}
\emph{Fix $\boldsymbol{p} \gg 0$, $0 < r < 1$, and $\theta  >0$. Then \newline (i) $\theta g(a ,\boldsymbol{p} ,r) =g(\theta a ,\theta \boldsymbol{p} ,r)$ and $\boldsymbol{x}^{ \ast }(a -g(a ,\boldsymbol{p} ,r) ,\boldsymbol{p}) =\boldsymbol{x}^{ \ast }(\theta a -g(\theta a ,\theta \boldsymbol{p} ,r) ,\theta \boldsymbol{p})$ for all $a$. 
\newline (ii) $\zeta _{i}(\theta \boldsymbol{p} ,r) =\zeta _{i}(\boldsymbol{p} ,r)$ for $1 \leq i \leq n$ and $\zeta _{n +1}(\theta \boldsymbol{p} ,r) =\theta \zeta _{n +1}(\boldsymbol{p} ,r)$. \newline Thus, if $(\boldsymbol{p} ,r)$ are equilibrium prices then $(\theta \boldsymbol{p} ,r)$ are also equilibrium prices. }

\begin{proof}
(i) Recall that a function $f(a ,\boldsymbol{p} ,r)$ is homogeneous of degree $l \geq 0$ in $(a ,\boldsymbol{p})$ if $f(\theta a ,\theta \boldsymbol{p} ,r) =\theta ^{l}f(a ,\boldsymbol{p} ,r)$ for all $\theta  >0$. We now show that $g$ is homogeneous of degree $1$ in $(a ,\boldsymbol{p})$. 

Fix $\boldsymbol{p} \gg 0$, $\theta  >0$, $a \in \mathbb{R}$, and $0 < r < 1$. 

Assume that $f(a,\boldsymbol{p} ,r)$ is homogeneous of degree $0$ in $(a ,\boldsymbol{p})$. We have 
\begin{align*}Tf(a ,\boldsymbol{p} ,r) &  =\max _{b \in C(a ,\boldsymbol{p})}\max _{\boldsymbol{x} \in X(a -b ,\boldsymbol{p})}U(\boldsymbol{x}) +\beta \sum _{\boldsymbol{y} \in \mathcal{Y}}e(\boldsymbol{y})f((1 +r)b +\boldsymbol{p} \cdot \boldsymbol{y} ,\boldsymbol{p} ,r) \\
 &  =\max _{\theta b \in C(\theta a ,\theta \boldsymbol{p})}\max _{\boldsymbol{x} \in X(\theta a -\theta b ,\theta \boldsymbol{p})}U(\boldsymbol{x}) +\beta \sum \limits _{\boldsymbol{y} \in \mathcal{Y}}^{\,}e(\boldsymbol{y})f((1 +r)\theta b +\theta \boldsymbol{p} \cdot \boldsymbol{y} ,\theta \boldsymbol{p} ,r) \\
 &  =\max _{z \in C(\theta a ,\theta \boldsymbol{p})}\max _{\boldsymbol{x} \in X(\theta a -z,\theta \boldsymbol{p})}U(\boldsymbol{x}) +\beta \sum \limits _{\boldsymbol{y} \in \mathcal{Y}}^{\,}e(\boldsymbol{y})f((1 +r)z +\theta \boldsymbol{p} \cdot \boldsymbol{y} ,\theta \boldsymbol{p} ,r) \\
 &  =Tf(\theta a ,\theta \boldsymbol{p} ,r).\end{align*}Thus, $Tf(a ,\boldsymbol{p} ,r)$ is homogeneous of degree $0$ in $(a ,\boldsymbol{p})$. The first and fourth equalities follow from the definition of $Tf$. The second equality follows from the facts that $X(a-b ,\boldsymbol{p}) =X(\theta a -\theta b ,\theta \boldsymbol{p})$, $b \in C(a ,\boldsymbol{p})$ if and only if $\theta b \in C(\theta a ,\theta \boldsymbol{p})$,\protect\footnote{ 
Recall that $C(a,\boldsymbol{p}) = [-\frac{ \min _{\boldsymbol{y} \in \mathcal{Y}}\boldsymbol{p} \cdot \boldsymbol{y}}{r},\min \left  \{a ,\sum _{i =1}^{n} \frac{ p_{i}\overline{b} } {(1-r)^{2} }\right \}]$.
} and $f((1 +r)\theta b +\theta \boldsymbol{p} \cdot \boldsymbol{y} ,\theta \boldsymbol{p} ,r) =f((1 +r)b +\boldsymbol{p} \cdot \boldsymbol{y} ,\boldsymbol{p} ,r)$ for $\theta  >0$.

We conclude that for all $n =1 ,2 ,3 \ldots$, $T^{n} f$ is homogeneous of degree $0$. From standard dynamic programming arguments,
$T^{n} f$ converges to $V$ uniformly. Since the set of functions that are homogeneous of degree $0$ is closed under uniform convergence, $V$ is homogeneous of degree zero. 

Let \begin{equation*}h(a,b ,\boldsymbol{p} ,r ,f) : =v(a -b ,\boldsymbol{p}) +\beta \sum \limits _{\boldsymbol{y} \in \mathcal{Y}}^{\,}e(\boldsymbol{y})f((1 +r)b +\boldsymbol{p} \cdot \boldsymbol{y} ,\boldsymbol{p} ,r) ,
\end{equation*}
where $v(a -b ,\boldsymbol{p}) =\max _{\boldsymbol{x} \in X(a -b ,\boldsymbol{p})}U(\boldsymbol{x})$. 

Note that $v$ is homogeneous of degree $0$ in $(a ,b ,\boldsymbol{p})$ since $X(a-b ,\boldsymbol{p}) =X(\theta a -\theta b ,\theta \boldsymbol{p})$. Thus, $h(a ,b ,\boldsymbol{p} ,r ,f)$ is homogeneous of degree $0$ in $(a ,b ,\boldsymbol{p})$ whenever $f$ is  homogeneous of degree $0$ in $(a ,\boldsymbol{p})$. Since $V$ is homogeneous of degree zero in $(a ,\boldsymbol{p})$ we conclude that  $h(a ,b ,\boldsymbol{p} ,r ,V)$ is homogeneous of degree $0$ in $(a ,b ,\boldsymbol{p})$. We have    
\begin{align*}h(\theta a ,\theta g(a ,\boldsymbol{p} ,r) ,\theta \boldsymbol{p} ,r ,V) &  =h(a ,g(a ,\boldsymbol{p} ,r) ,\boldsymbol{p} ,r ,V) \\
 &  =\max _{b \in C(a ,\boldsymbol{p})}h(a ,b ,\boldsymbol{p} ,r ,V) \\
 &  =v(a -\boldsymbol{p} ,r) \\
 &  =V(\theta a ,\theta \boldsymbol{p} ,r) \\
 &  =h(\theta a ,g(\theta a ,\theta \boldsymbol{p} ,r) ,\theta \boldsymbol{p} ,r ,V) .\end{align*}The single-valuedness of the savings policy function $g$ (see Lemma \ref{Lemma 1} below) implies that $g(\theta a ,\theta \boldsymbol{p} ,r) =\theta g(a ,\boldsymbol{p} ,r)$. We conclude that $g$ is homogeneous of degree $1$ in $(a ,\boldsymbol{p})$. 

The following equalities show that the demand function $\boldsymbol{x}^{ \ast }$ is homogeneous of degree $0$ in $(a ,\boldsymbol{p})$:   
\begin{align*}\boldsymbol{x}^{ \ast }(a -g(a ,\boldsymbol{p} ,r) ,\boldsymbol{p}) &  =\ensuremath{\operatorname*{argmax}}_{\boldsymbol{x} \in X(a -g(a ,\boldsymbol{p} ,r) ,\boldsymbol{p})}U(\boldsymbol{x}) \\
 &  =\ensuremath{\operatorname*{argmax}}_{\boldsymbol{x} \in X(\theta a -\theta g(a ,\boldsymbol{p} ,r) ,\theta \boldsymbol{p})}U(\boldsymbol{x}) \\
 &  =\ensuremath{\operatorname*{argmax}}_{\boldsymbol{x} \in X(\theta a - g(\theta a ,\theta \boldsymbol{p} ,r) ,\theta \boldsymbol{p})}U(\boldsymbol{x}) \\
 &  =\boldsymbol{x}^{ \ast }\left (\theta a -g(\theta a ,\theta \boldsymbol{p} ,r) ,\theta \boldsymbol{p}\right ) .\end{align*}

(ii) Let $0<r<1$ be fixed. Then $A(\boldsymbol{p},r)$ is compact for all $\boldsymbol{p} \gg 0$. We start with the following definitions: 

We say that a function $f(a ,\boldsymbol{p})$ is bounded in $a$ if for every $\boldsymbol{p} \gg 0$, there exists an  $M>0$ such that   $\vert f(a,\boldsymbol{p}) | \leq M$ for all $a \in A(\boldsymbol{p},r)$.  

We say that a probability measure $\lambda ( \cdot ;\boldsymbol{p} ,r) \in \mathcal{P}(A(\boldsymbol{p} ,r))$ is homogeneous of degree $l \geq 0$ in $\boldsymbol{p}$ if for every continuous function $f(a ,\boldsymbol{p})$ that is homogeneous of degree $l$ in $(a ,\boldsymbol{p})$, and bounded in $a$, and all $\theta >0$ we have \begin{equation}\int f(a ,\theta \boldsymbol{p})\lambda (da;\theta \boldsymbol{p} ,r) =\theta ^{l}\int f(a ,\boldsymbol{p})\lambda (da;\boldsymbol{p} ,r). \label{homog}
\end{equation} 
We now show that the invariant wealth distribution $\mu $ is homogeneous of degree $l \geq 0$ in $\boldsymbol{p}$. 

Assume that the probability measure $\lambda ( \cdot ;\boldsymbol{p} ,r) \in \mathcal{P}(A(\boldsymbol{p} ,r))$ is homogeneous of degree $l$ in $\boldsymbol{p}$. Let $f(a ,\boldsymbol{p})$ be a continuous function that is homogeneous of degree $l$ in $(a ,\boldsymbol{p})$ and bounded in $a$, and let $\theta  >0$, $\boldsymbol{p} \gg 0$, and $0<r<1$.  We have
\begin{align*}\theta ^{l}\int f(a ,\boldsymbol{p})M\lambda (da;\boldsymbol{p} ,r) &  =\theta ^{l}\int \sum \limits _{\boldsymbol{y} \in \mathcal{Y}}^{\,}e(\boldsymbol{y})f((1 +r)g(a ,\boldsymbol{p} ,r) +\boldsymbol{p} \cdot \boldsymbol{y} ,\boldsymbol{p})\lambda (da;\boldsymbol{p} ,r) \\
 &  =\int \sum \limits _{\boldsymbol{y} \in \mathcal{Y}}^{\,}e(\boldsymbol{y})f((1 +r)g(a ,\theta \boldsymbol{p} ,r) +\theta \boldsymbol{p} \cdot \boldsymbol{y} ,\theta \boldsymbol{p})\lambda (da;\theta \boldsymbol{p} ,r) \\
 &  =\int f(a ,\theta \boldsymbol{p})M\lambda (da;\theta \boldsymbol{p} ,r) .\end{align*} 
 The first and last equalities follow from Equation (\ref{eq:stationary}) (see Lemma \ref{Unique stationary is contin} below). The second equality follows from the facts that $\widetilde{f}(a ,\boldsymbol{p}) : =$$\sum \limits _{\boldsymbol{y} \in \mathcal{Y}}^{\,}e(\boldsymbol{y})f((1 +r)g(a ,\boldsymbol{p} ,r) +\boldsymbol{p} \cdot \boldsymbol{y} ,\boldsymbol{p})$ is homogeneous of degree $l$ in $(a ,\boldsymbol{p})$ and  $\lambda $ is homogeneous of degree $l$ in $\boldsymbol{p}$. We conclude that for all $k$, $M^{k}\lambda $ is homogeneous of degree $l$ in $\boldsymbol{p}$. 


From Lemma \ref{unique stationary dist} (see below), $M^{k}\lambda $ converges weakly to $\mu$ for all $(\boldsymbol{p} ,r)$ such that $\boldsymbol{p} \gg 0$, and $0 < r <1$. For every continuous function $f(a ,\boldsymbol{p})$ that is homogeneous of degree $l$ in $(a ,\boldsymbol{p})$ and bounded in $a$, we have 
\begin{align*}\int f(a ,\theta \boldsymbol{p})\mu (da;\theta \boldsymbol{p} ,r) &  =\lim _{k \rightarrow \infty }\int f(a,\theta \boldsymbol{p})M^{k}\lambda (da;\theta \boldsymbol{p} ,r) \\
 &  =\lim _{k \rightarrow \infty }\theta ^{l}\int f(a ,\boldsymbol{p})M^{k}\lambda (da;\boldsymbol{p} ,r) \\
 &  =\theta ^{l}\int f(a ,\boldsymbol{p})\mu (da;\boldsymbol{p} ,r) .\end{align*}
 Thus, $\mu $ is homogeneous of degree $l$ in $\boldsymbol{p}$.  

The facts that $g(a,\boldsymbol{p} ,r)$ is a continuous function on $\mathbb{R} \times \mathbb{R}_{+ +}^{n} \times (0,1)$ (see Lemma \ref{Lemma 1}) and that $A(\boldsymbol{p},r)$ is compact for all $\boldsymbol{p} \gg 0$, $r \in (0,1)$ imply that $g$ is bounded in $a$. We established in part (i) that $g$ is homogeneous of degree $1$ in $(a ,\boldsymbol{p})$. Thus, using the fact that $\mu $ is homogeneous of degree $1$ in $\boldsymbol{p}$ yields 
\begin{equation*}\zeta _{n +1}(\theta \boldsymbol{p} ,r) =\int g(a ,\theta \boldsymbol{p} ,r)\mu (da;\theta \boldsymbol{p} ,r) =\theta \int  g(a ,\boldsymbol{p} ,r)\mu (da;\boldsymbol{p} ,r) =\theta \zeta _{n +1}(\boldsymbol{p} ,r).
\end{equation*} 
Similarly, in part (i) we established that  for all $1 \leq i \leq n$ the demand  function $x_{i}^{ \ast }(a -g(a ,\boldsymbol{p} ,r) ,\boldsymbol{p})$ is homogeneous of degree $0$ in $(a ,\boldsymbol{p})$. The demand function is continuous (see Lemma \ref{Lemma 1} below) and  bounded in $a$. Hence, using the fact that $\mu $ is homogeneous of degree $0$ in $\boldsymbol{p}$ yields 
\begin{align*}\zeta _{i}(\theta \boldsymbol{p} ,r) &  =\int x_{i}^{ \ast }(a -g(a ,\theta \boldsymbol{p} ,r) ,\theta \boldsymbol{p})\mu (da;\theta \boldsymbol{p} ,r) -\sum \limits _{y_{i} \in \mathcal{Y}_{i}}^{\,}e_{i}(y_{i})y_{i} \\
 &  =\int x_{i}^{ \ast }(a -g(a ,\boldsymbol{p} ,r) ,\boldsymbol{p})\mu (da;\boldsymbol{p} ,r) -\sum \limits _{y_{i} \in \mathcal{Y}_{i}}^{\,}e_{i}(y_{i})y_{i} =\zeta _{i}(\boldsymbol{p} ,r) .\end{align*} Thus, if $(\boldsymbol{p} ,r)$ are equilibrium prices, i.e., $\zeta (\boldsymbol{p} ,r) =0$, then $\zeta (\theta \boldsymbol{p} ,r) =0$; and so $(\theta \boldsymbol{p} ,r)$ are also equilibrium prices. 
\end{proof}

\subsection{The existence of a competitive stationary equilibrium}
In this section we prove the existence of a competitive stationary equilibrium. 

\noindent \textbf{Theorem \ref{Theorem exist}.}
\emph{Suppose that Assumption \ref{Assumption 0} holds. Then, there exists a competitive stationary equilibrium. }

Recall that the sets $\Lambda$ and $\boldsymbol{P}$ are given by  $\Lambda  =\{(\boldsymbol{p} ,r) \in \mathbb{R}_{ +}^{n} \times \mathbb{R}_{ +} :\sum _{i =1}^{n}p_{i} +r = 1\}$ and
\begin{equation*}\boldsymbol{P} =\{(\boldsymbol{p} ,r) \in \Lambda  :\boldsymbol{p} \gg 0 ,r >0\} . 
\end{equation*}
The excess demand function $\zeta  :\boldsymbol{P} \rightarrow \mathbb{R}^{n +1}$ is given by
\begin{equation*}
\zeta (\boldsymbol{p} ,r) =(\zeta _{1}(\boldsymbol{p} ,r) , \ldots ,\zeta _{n}(\boldsymbol{p} ,r) ,\zeta _{n +1}(\boldsymbol{p} ,r))\end{equation*} 
where for $i =1 , \ldots ,n$,
\begin{equation*}\zeta _{i}(\boldsymbol{p} ,r) =\int _{A(\boldsymbol{p},r)}x_{i}^{ \ast }(a -g(a ,\boldsymbol{p} ,r) ,\boldsymbol{p})\mu (da;\boldsymbol{p} ,r) -\sum \limits _{y_{i} \in \mathcal{Y}_{i}}^{\,}e_{i}(y_{i})y_{i}
\end{equation*}
is the excess demand for good $i$, and \begin{equation*}
 \zeta _{n +1}(\boldsymbol{p} ,r) = -\int _{A(\boldsymbol{p} ,r)}g(a ,\boldsymbol{p} ,r)\mu (da;\boldsymbol{p} ,r) \end{equation*}
 is the excess demand for savings. 
 Note that if $\zeta (\boldsymbol{p} ,r) =\boldsymbol{0}$ then $(\boldsymbol{p},r)$ are equilibrium prices, $\mu ( \cdot ;\boldsymbol{p} ,r)$ is the equilibrium wealth distribution, $\boldsymbol{x}^{ \ast }(a -g(a ,\boldsymbol{p} ,r) ,\boldsymbol{p})$ is the equilibrium demand function, and $g(a ,\boldsymbol{p} ,r)$ is the equilibrium savings policy function.

In the next Proposition we show that the excess demand function is continuous, satisfies Walras' law and some boundary and boundness conditions. For $x \in \mathbb{R}^{n}$ we write $\left \Vert x\right \Vert _{1} =\sum \limits _{j =1}^{n}\vert x_{i}\vert $.

\begin{proposition}
\label{Prop: existence} The excess demand function $\zeta  :\boldsymbol{P} \rightarrow \mathbb{R}^{n +1}$ satisfies the following properties. 

(i) The function $\zeta $ is continuous. 

(ii) The function $\zeta $ satisfies Walras' law, i.e., $(\boldsymbol{p} ,r) \cdot \zeta (\boldsymbol{p} ,r) =0$ for all $(\boldsymbol{p} ,r) \in \boldsymbol{P}$.

(iii) If  $(\boldsymbol{p}_{q} ,r_{q}) \rightarrow (\boldsymbol{p} ,r) =(p_{1}, \ldots ,p_{n} ,r) \in \Lambda \backslash \boldsymbol{P}$ with $\{\boldsymbol{p}_{q} ,r_{q}\} \subseteq \boldsymbol{P}$ and $p_{k}>0$ for some $1\leq k \leq n$, then $\lim _{q \rightarrow \infty }\left \Vert \zeta (\boldsymbol{p}_{q} ,r_{q})\right \Vert _{1} =\infty $. 

(iv) If $\{\boldsymbol{p}_{q} ,r_{q}\} \subseteq \boldsymbol{P}$, $(\boldsymbol{p}_{q} ,r_{q}) \rightarrow (\boldsymbol{p} ,r) =(p_{1} ,\ldots  ,p_{n} ,r)$ and $p_{k} >0$, then the sequence $\{\zeta _{k} (\boldsymbol{p}_{q} ,r_{q})\}$ of the $k^{t h}$ components of $\{\zeta  (\boldsymbol{p}_{q} ,r_{q})\}$ is bounded. Similarly, $r \in (0,1)$ implies that the sequence $\{\zeta _{n +1} (\boldsymbol{p}_{q} ,r_{q})\}$ is bounded. 

(v)  There exists $\xi  >0$ such that $\zeta _{i} (\boldsymbol{p} ,r) \geq  -\xi $ for all $1 \leq i \leq n$ and all $(\boldsymbol{p} ,r) \in \boldsymbol{P}$, and $\zeta _{n+1} (\boldsymbol{p} ,r) \geq  -\xi $ for all $(\boldsymbol{p} ,r) \in \boldsymbol{P}$ such that $r \leq \delta < 1$ for some $\delta \in (0,1)$. 

\end{proposition}

\begin{remark} \label{Remark:Boundary}
Note that the standard boundary conditions that imply the existence of a competitive equilibrium in the static Arrow-Debreu model (see \cite{debreu1982existence} and \cite{hildenbrand2014introduction}) are not satisfied in our setting. This is because conditions (iii) of Proposition \ref{Prop: existence} holds only for sequences $\{\boldsymbol{p}_{q} ,r_{q}\} $ such that $(\boldsymbol{p}_{q} ,r_{q}) \rightarrow (\boldsymbol{p} ,r) =(p_{1}, \ldots ,p_{n} ,r) \in \Lambda \backslash \boldsymbol{P}$ with $p_{k}>0$ for some $1\leq k \leq n$ and condition (v) does not imply that the excess demand is bounded over  $\boldsymbol{P}$ as required in \cite{debreu1982existence}). In proving the existence of a CSE, we are able to overcome this difficulty by using the properties of the savings policy function. 
\end{remark}

The idea of the proof of Theorem \ref{Theorem exist} is to show that the conditions of Proposition \ref{Prop: existence} hold and then to prove that these conditions together with the properties of savings policy function imply the existence of a CSE. 
Lemmas \ref{Lemma 1} and \ref{unique stationary dist} show that the excess demand function $\zeta (\boldsymbol{p} ,r)$ is a well-defined function. In Lemmas \ref{Unique stationary is contin} and \ref{Lemma ED is contin} we prove that the excess demand function is continuous. In Lemma \ref{Lemma ED Walras} we prove that the excess demand function satisfies Walras' law. In Lemma \ref{Lemma Bdd from below} and Lemma \ref{Lemma boundary} we prove that the excess demand function satisfies the boundness and boundary conditions (conditions (iii), (iv) and (v) of Proposition \ref{Prop: existence}).

\begin{lemma}
\label{Lemma 1}The savings policy function $g(a ,\boldsymbol{p} ,r)$ is single-valued and continuous in $(a ,\boldsymbol{p} ,r)$ on $\mathbb{R} \times \mathbb{R}^{n}_{++} \times (0,1)$. The value function $V(a ,\boldsymbol{p} ,r)$ is continuous in $(a ,\boldsymbol{p} ,r)$ on $\mathbb{R} \times \mathbb{R}^{n}_{++} \times(0,1)$, increasing in $a$, and strictly concave in $a$.
\end{lemma}

\begin{proof}
Note that $v(a- b,\boldsymbol{p}) =\max _{\boldsymbol{x} \in X(a -b ,\boldsymbol{p})}U(\boldsymbol{x})$ 
is strictly concave in $(a,b)$. To see this, let $(a_{1} ,b_{1}) \in \mathbb{R}^{2}$, $(a_{2} ,b_{2}) \in \mathbb{R}^{2}$, $\gamma  \in [0 ,1]$, $a_{\gamma } =\gamma a_{1} +(1 -\gamma )a_{2}$, and $b_{\gamma } =\gamma b_{1} +(1 -\gamma )b_{2}$. 

We have
\begin{align*}v(a_{\gamma } - b_{\gamma } ,\boldsymbol{p}) &  =\max _{\boldsymbol{x} \in X(a_{\gamma } -b_{\gamma } ,\boldsymbol{p})}U(\boldsymbol{x}) \\
 &  \geq U(\gamma \boldsymbol{x}^{ \ast }(a_{1} -b_{1} ,\boldsymbol{p}) +(1 -\gamma )\boldsymbol{x}^{ \ast }(a_{2} -b_{2} ,\boldsymbol{p})) \\
 &  >\gamma U(\boldsymbol{x}^{ \ast }(a_{1} -b_{1} ,\boldsymbol{p})) +(1 -\gamma )U(\boldsymbol{x}^{ \ast }(a_{2} -b_{2} ,\boldsymbol{p})) \\
 &  =\gamma v(a_{1} -b_{1} ,\boldsymbol{p}) +(1 -\gamma )v(a_{2} -b_{2} ,\boldsymbol{p}) .\text{}\end{align*} The first inequality follows from the fact that 
 $\gamma \boldsymbol{x}^{ \ast }(a_{1} -b_{1} ,\boldsymbol{p}) +(1 -\gamma )\boldsymbol{x}^{ \ast }(a_{2} -b_{2} ,\boldsymbol{p}) \in X(a_{\gamma } -b_{\gamma } ,\boldsymbol{p})$. The second inequality follows from the fact that $U$ is strictly concave. We conclude that $v$ is strictly concave in $(a ,b)$. Furthermore, since $U$ is continuous and $X(a -b ,\boldsymbol{p})$ is a continuous correspondence, i.e., $X$ is upper hemicontinuous and lower hemicontinuous, the maximum theorem (see Theorem 17.31 in \cite{aliprantis2006infinite}) implies that $v(a -b ,\boldsymbol{p})$ is continuous. Because $U$ is an increasing function, $v$ is increasing in $a$. Now standard dynamic programming arguments show that $g(a ,\boldsymbol{p} ,r)$ is single-valued and continuous in $(a ,\boldsymbol{p} ,r)$ and that $V(a,\boldsymbol{p} ,r)$ is continuous, as well as strictly concave and increasing in $a$ (see Chapter 9 in \cite{stokey1989}). 
\end{proof}

\begin{lemma}
\label{unique stationary dist}For every $(\boldsymbol{p} ,r) \in \boldsymbol{P}$ there exists a unique invariant wealth distribution $\mu ( \cdot ;\boldsymbol{p} ,r) \in \mathcal{P}(A(\boldsymbol{p} ,r))$. Furthermore, for all $\lambda ( \cdot ;\boldsymbol{p} ,r) \in \mathcal{P}(A(\boldsymbol{p},r))$, the sequence of measures $\{M^{n}\lambda \}$ converges weakly to $\mu ( \cdot ;\boldsymbol{p} ,r) \in \mathcal{P}(A(\boldsymbol{p} ,r))$.
\end{lemma}

\begin{proof}
Fix $(\boldsymbol{p} ,r) \in \boldsymbol{P}$. Define the Markov chain 
\begin{equation*}
 Q(a,D) =\sum _{\boldsymbol{y}  \in \mathcal{Y}}e(\boldsymbol{y})1_{D}((1 +r)g(a ,\boldsymbol{p} ,r) +\boldsymbol{p} \cdot \boldsymbol{y}). \end{equation*} for any  $D \in \mathcal{B}(A(\boldsymbol{p},r))$ where $1_{D}$ is the indicator function of the set $D \in \mathcal{B}(A(\boldsymbol{p},r))$. 
 
 We prove a more general result than the result stated in Lemma \ref{unique stationary dist}: we show that the Markov chain $Q$ is uniformly ergodic.\protect\footnote{
Recall that the Markov chain $Q$ is called uniformly ergodic if it has an invariant distribution $\mu $ and $\sup _{D \in \mathcal{B}(A(\boldsymbol{p} ,r))}\vert Q^{n}(a ,D) -\mu (D)\vert  \leq M\rho ^{n}$ for some $\rho  <1$, $M <\infty $ and for all $n \in \mathbb{N}$, $a \in A(\boldsymbol{p},r)$. Clearly, if $Q$ is uniformly ergodic then Lemma \ref{unique stationary dist} holds.
} The proof follows a similar line to the proofs in \cite{rabault2002borrowing} and in \cite{benhabib2015wealth}, so we only provide a sketch of the proof.\protect\footnote{
See also \cite{schechtman1977some}, \cite{ma2020income}, and \cite{foss2018stochastic}.}

The Markov chain $Q$ is said to satisfy the Doeblin condition if there exists a positive integer $n_{0}$, $\epsilon  >0$ and a probability measure $\upsilon $ on $A(\boldsymbol{p},r)$ such that $Q^{n_{0}}(a ,D) \geq \epsilon \upsilon (D)$ for all $a \in A(\boldsymbol{p},r)$ and all $D \in \mathcal{B}(A(\boldsymbol{p},r))$.

Assume that $(1+r)\beta <1$. Under Assumption \ref{Assumption 0}, the arguments in Proposition 3.1 in \cite{rabault2002borrowing} show that the borrowing constraint binds with positive probability after a finite number of periods for any initial wealth level $a \in A(\boldsymbol{p},r)$. In other words, for any initial wealth level $a \in A(\boldsymbol{p},r)$, we have $g(a ,\boldsymbol{p} ,r) =\underline{b}$ with a positive probability after a finite number of periods. Thus, if we define the probability measure $\upsilon (D) =\sum _{\boldsymbol{y} \in \mathcal{Y}}e(\boldsymbol{y})1_{D}((1 +r)\underline{b} +\boldsymbol{p} \cdot \boldsymbol{y})$, we can find a positive integer $n_{0}$ and $\epsilon  >0$ such that $Q^{n_{0}}(a ,D) \geq \epsilon \upsilon (D)$ for all $a \in A(\boldsymbol{p},r)$ and all $D \in \mathcal{B}(A(\boldsymbol{p},r))$. We conclude that $Q$ satisfies the Doeblin condition. From the facts that $M :\mathcal{P}(A(\boldsymbol{p},r)) \rightarrow \mathcal{P}(A(\boldsymbol{p},r))$ is continuous (see a more general result in Lemma \ref{Unique stationary is contin}) and  $\mathcal{P}(A(\boldsymbol{p},r))$ is compact in the weak topology (since $A(\boldsymbol{p},r)$ is compact), Schauder's fixed-point theorem (see Corollary 17.56 in \cite{aliprantis2006infinite}) implies that $M$ has at least one fixed point. That is, $Q$ has at least one invariant distribution. A Markov chain that has an invariant distribution and satisfies the Doeblin condition is uniformly ergodic (see Theorem 8 in \cite{roberts2004general}). 

If $(1+r)\beta \geq 1$ then the results in \cite{chamberlain2000optimal} show that the upper bound on savings binds with positive probability after a finite number of periods for any initial state. A similar argument to the argument above proves that $Q$ is uniformly ergodic. This completes the proof the Lemma.   
\end{proof}

We say that $w_{n} :\mathbb{R} \rightarrow \mathbb{R}$ converges continuously to $w$ if $w_{n} (a_{n}) \rightarrow w (a)$ whenever $a_{n} \rightarrow a$. Lemma \ref{serfozo} provides a  bounded convergence theorem with varying measures. For a proof, see Theorem 3.3 in \cite{serfozo1982convergence}.\footnote{See \cite{feinberg2019fatou} for a more general result of this type.} We will use this Lemma to prove the continuity of the excess demand function.

\begin{lemma}
\label{serfozo}Assume that $w_{n} :\mathbb{R} \rightarrow \mathbb{R}$ is a uniformly bounded sequence of functions. If $w_{n} :\mathbb{R} \rightarrow \mathbb{R}$ converges continuously to $w$ and $\lambda _{n} \in \mathcal{P}(\mathbb{R})$ converges weakly to $\lambda  \in \mathcal{P}(\mathbb{R})$ then
\begin{equation*}\underset{n \rightarrow \infty }{\lim }\int w_{n} (a) \lambda _{n} (d a) =\int w (a) \lambda  (d a)\text{.}
\end{equation*}
\end{lemma}

\begin{lemma}
\label{Unique stationary is contin}The unique invariant wealth distribution $\mu $ is continuous in $(\boldsymbol{p} ,r)$ on $\boldsymbol{P}$, i.e., if $\{\boldsymbol{p}_{n} ,r_{n}\}$ converges to $(\boldsymbol{p} ,r)$, then $\mu (\boldsymbol{p}_{n} ,r_{n})$ converges weakly to $\mu (\boldsymbol{p} ,r)$. 
\end{lemma}

\begin{proof}
First note that for every bounded and measurable function $f :\mathbb{R} \rightarrow \mathbb{R}$ and for all $(\boldsymbol{p} ,r)$ such that $\boldsymbol{p} \gg 0$ and $r \in (0,1)$  we have
\begin{equation}\int f(a)M\lambda (da;\boldsymbol{p} ,r) =\int \sum \limits _{\boldsymbol{y} \in \mathcal{Y}}e(\boldsymbol{y})f((1 +r)g(a ,\boldsymbol{p} ,r) +\boldsymbol{p} \cdot \boldsymbol{y})\lambda (da;\boldsymbol{p} ,r). \label{eq:stationary}
\end{equation}To see this, note that if $f =1_{D}$ then Equality (\ref{eq:stationary}) holds from the definition of $M$. A standard argument shows that Equality (\ref{eq:stationary}) holds for any bounded and measurable $f$. 

Assume that $\{\boldsymbol{p}_{n} ,r_{n}\} \subseteq \boldsymbol{P}$ converges to $(\boldsymbol{p},r)\in\boldsymbol{P}$. Let $\{\mu(\boldsymbol{p}_{k} ,r_{k})\}$ be a subsequence of $\{\mu (\boldsymbol{p}_{n} ,r_{n})\}$ that converges weakly to $\overline{\mu }(\boldsymbol{p} ,r)$. Let $f :\mathbb{R} \rightarrow \mathbb{R}$ be a bounded and continuous function. From the continuity of the savings policy function $g$, we have \begin{equation*}\lim _{k\rightarrow \infty }f((1 +r_{k})g(a_{k} ,\boldsymbol{p}_{k} ,r_{k}) +\boldsymbol{p}_{k} \cdot \boldsymbol{y}) =f((1 +r)g(a ,\boldsymbol{p} ,r) +\boldsymbol{p} \cdot \boldsymbol{y})
\end{equation*} whenever $\lim_{k\rightarrow\infty} (a_{k} ,\boldsymbol{p}_{k} ,r_{k}) =(a ,\boldsymbol{p} ,r)$.

Let us define $$w_{k}(a) =\sum \limits _{\boldsymbol{y} \in \mathcal{Y}}^{\,}e(\boldsymbol{y})f((1 +r_{k})g(a ,\boldsymbol{p}_{k} ,r_{k}) +\boldsymbol{p}_{k} \cdot \boldsymbol{y}) \text{  and  } w(a) =\sum \limits _{\boldsymbol{y} \in \mathcal{Y}}^{\,}e(\boldsymbol{y})f((1 +r)g(a ,\boldsymbol{p} ,r) +\boldsymbol{p} \cdot \boldsymbol{y}).$$ Then, $w_{k}(a)$ is a uniformly bounded sequence of functions that converges continuously to $w(a)$.

Applying Lemma \ref{serfozo} and using Equality (\ref{eq:stationary}) twice yield
\begin{align*}\lim _{k \rightarrow \infty }\int f(a)\mu (da;\boldsymbol{p}_{k} ,r_{k}) &  =\lim _{k \rightarrow \infty }\int \sum \limits _{\boldsymbol{y} \in \mathcal{Y}}^{\,}e(\boldsymbol{y})f((1 +r_{k})g(a ,\boldsymbol{p}_{k} ,r_{k}) +\boldsymbol{p}_{k} \cdot \boldsymbol{y})\mu (da;\boldsymbol{p}_{k} ,r_{k}) \\
 &  =\lim _{k \rightarrow \infty }\int w_{k}(a)\mu (da;\boldsymbol{p}_{k} ,r_{k}) \\
 &  =\int w(a)\overline{\mu }(da;\boldsymbol{p} ,r) \\
 &  =\int f(a)M\overline{\mu }(da;\boldsymbol{p} ,r) .\end{align*}
 Because $\mu (\boldsymbol{p}_{k} ,r_{k})$ converges weakly to $\overline{\mu }(\boldsymbol{p} ,r)$, we also have
 \begin{equation*} \lim _{k \rightarrow \infty }\int f(a)\mu (da;\boldsymbol{p}_{k} ,r_{k}) =\int f(a)\overline{\mu }(da;\boldsymbol{p} ,r).
 \end{equation*}
 Thus,  \begin{equation*}\int f(a)M\overline{\mu }(da;\boldsymbol{p} ,r) =\int f(a)\overline{\mu }(da;\boldsymbol{p} ,r)
\end{equation*}
which implies that $\overline{\mu } =M\overline{\mu }$, since $\overline{\mu }$ and $M\overline{\mu }$ are Borel probability measures that agree on all open sets. From Lemma \ref{unique stationary dist}, $\mu $ is the unique fixed point of $M$, and thus,  $\overline{\mu } =\mu $. 

We conclude that each subsequence of $\{\mu (\boldsymbol{p}_{n} ,r_{n})\}$ that converges weakly at all converges weakly to $\mu (\boldsymbol{p} ,r)$. Furthermore, since $A(\boldsymbol{p},r)$ is compact, for all $(\boldsymbol{p},r) \in \boldsymbol{P}$ we can assume that the supports of $\mu (\boldsymbol{p}_{n} ,r_{n})$ and $\mu (\boldsymbol{p},r)$ are contained in a compact set so the sequence $\{\mu (\boldsymbol{p}_{n} ,r_{n})\}$ is a tight sequence of probability measures. Thus, $\mu (\boldsymbol{p}_{n} ,r_{n})$ converges weakly to $\mu (\boldsymbol{p} ,r)$ (see the Corollary after Theorem 25.10 in \cite{billingsley2008probability}). 
\end{proof}

\begin{lemma}
\label{Lemma ED is contin} The excess demand function $\zeta (\boldsymbol{p} ,r)$ is continuous on $\boldsymbol{P}$. 
\end{lemma}

\begin{proof}
Assume that the sequence $\{\boldsymbol{p}_{n} ,r_{n}\} \subseteq \boldsymbol{P}$ converges to $(\boldsymbol{p} ,r) \in \boldsymbol{P}$. Fix $i$ such that $1 \leq i \leq n$. Define $w_{n}(a) =x_{i}^{ \ast }(a -g(a ,\boldsymbol{p}_{n} ,r_{n}) ,\boldsymbol{p}_{n})$ and $w(a) =x_{i}^{ \ast }(a -g(a ,\boldsymbol{p} ,r) ,\boldsymbol{p})$. The continuity of $x_{i}^{ \ast }$ and of $g$ imply that $w_{n}$ converges continuously to $w$, i.e., $w_{n}(a_{n}) \rightarrow w(a)$ whenever $a_{n} \rightarrow a$. The sequence of functions $\{w_{n}(a)\}$ is bounded (see Lemma \ref{Lemma boundary}). Using Lemma \ref{serfozo} and  the fact that $\mu (\boldsymbol{p}_{n} ,r_{n})$ converges weakly to $\mu (\boldsymbol{p} ,r)$ (see Lemma \ref{Unique stationary is contin}) yield 
\begin{align*}\lim _{n \rightarrow \infty }\zeta _{i}(\boldsymbol{p}_{n} ,r_{n}) &  =\lim _{n \rightarrow \infty }\int w_{n}(a)\mu (da;\boldsymbol{p}_{n} ,r_{n}) -\sum _{y_{i} \in \mathcal{Y}_{i}}e_{i}(y_{i})y_{i} \\
 &  =\int w(a)\mu (da;\boldsymbol{p} ,r) -\sum _{y_{i} \in \mathcal{Y}_{i}}e_{i}(y_{i})y_{i} \\
 &  =\zeta _{i}(\boldsymbol{p} ,r) .\end{align*}
 Thus, $\zeta _{i}(\boldsymbol{p} ,r)$ is continuous for $1 \leq i \leq n$. A similar argument shows that $\zeta _{n +1}(\boldsymbol{p} ,r)$ is continuous. We conclude that $\zeta (\boldsymbol{p} ,r)$ is continuous.
\end{proof}

\begin{lemma} \label{Lemma ED Walras}
The excess demand function $\zeta (\boldsymbol{p} ,r)$ satisfies Walras' law, i.e., $(\boldsymbol{p} ,r) \cdot \zeta (\boldsymbol{p} ,r) =0$ for all $(\boldsymbol{p} ,r) \in \boldsymbol{P}$.
\end{lemma}

\begin{proof}
Fix $(\boldsymbol{p} ,r) \in \boldsymbol{P}$. Equation (\ref{eq:stationary}) implies that
\begin{align*}\int a\mu (da;\boldsymbol{p} ,r) &  =\int \sum \limits _{\boldsymbol{y} \in \mathcal{Y}}^{\,}e(\boldsymbol{y})((1 +r)g(a ,\boldsymbol{p} ,r) +\boldsymbol{p} \cdot \boldsymbol{y})\mu (da ;\boldsymbol{p} ,r) \\
 &  =(1 +r)\int g(a ,\boldsymbol{p} ,r)\mu (da;\boldsymbol{p} ,r) +\sum \limits _{\boldsymbol{y} \in \mathcal{Y}}^{\,}e(\boldsymbol{y})\boldsymbol{p} \cdot \boldsymbol{y} .\end{align*}Note that $\sum \limits _{\boldsymbol{y} \in \mathcal{Y}}^{\,}e(\boldsymbol{y})\boldsymbol{p} \cdot \boldsymbol{y} =\sum \limits _{i =1}^{n}p_{i}\sum _{y_{i} \in \mathcal{Y}_{i}}e_{i}(y_{i})y_{i}$. To see this, let $\mathcal{Y} =\{\boldsymbol{y}^{1} ,\ldots  ,\boldsymbol{y}^{l}\}$ and reason as follows:  \begin{equation*}\sum \limits _{\boldsymbol{y} \in \mathcal{Y}}^{\,}e(\boldsymbol{y})\boldsymbol{p} \cdot \boldsymbol{y} =e(\boldsymbol{y}^{1})\boldsymbol{p} \cdot \boldsymbol{y}^{1} +\ldots  +e(\boldsymbol{y}^{l})\boldsymbol{p} \cdot \boldsymbol{y}^{l} =\sum \limits _{i =1}^{n}p_{i}\sum \limits _{j =1}^{l}e(\boldsymbol{y}^{j})y_{i}^{j} =\sum \limits _{i =1}^{n}p_{i}\sum _{y_{i} \in \mathcal{Y}_{i}}e_{i}(y_{i})y_{i} .
\end{equation*}

From the agents' budget constraints, we have $\boldsymbol{p} \cdot \boldsymbol{x}^{ \ast }(a -g(a ,\boldsymbol{p} ,r) ,\boldsymbol{p}) =a -g(a ,\boldsymbol{p} ,r)$. 

The last equation implies \begin{equation*} \sum \limits _{i =1}^{n}p_{i}\int x_{i}^{ \ast }(a -g(a ,\boldsymbol{p} ,r) ,\boldsymbol{p})\mu (da;\boldsymbol{p} ,r) =\int (a -g(a ,\boldsymbol{p} ,r))\mu (da;\boldsymbol{p} ,r). 
\end{equation*}

Thus,
\begin{align*}(\boldsymbol{p} ,r) \cdot \zeta (\boldsymbol{p} ,r) &  =\sum \limits _{i =1}^{n}p_{i}\int x_{i}^{ \ast }(a -g(a ,\boldsymbol{p} ,r) ,\boldsymbol{p})\mu (da;\boldsymbol{p} ,r) -\sum \limits _{i =1}^{n}p_{i}\sum _{y_{i} \in \mathcal{Y}_{i}}e_{i}(y_{i})y_{i} -r\int g(a ,\boldsymbol{p} ,r)\mu (da;\boldsymbol{p} ,r) \\
 &  =\int (a -g(a ,\boldsymbol{p} ,r))\mu (da;\boldsymbol{p} ,r) -\sum \limits _{i =1}^{n}p_{i}\sum _{y_{i} \in \mathcal{Y}_{i}}e_{i}(y_{i})y_{i} -r\int g(a ,\boldsymbol{p} ,r)\mu (da;\boldsymbol{p} ,r) \\
 &  =\int a\mu (da;\boldsymbol{p} ,r) -(1 +r)\int g(a ,\boldsymbol{p} ,r)\mu (da;\boldsymbol{p} ,r) -\sum \limits _{\boldsymbol{y} \in \mathcal{Y}}^{\,}e(\boldsymbol{y})\boldsymbol{p} \cdot \boldsymbol{y} =0,\end{align*}
 which proves that $\zeta (\boldsymbol{p} ,r)$ satisfies Walras' law.
\end{proof}

\begin{lemma}
\label{Lemma Bdd from below} There exists $\xi  >0$ such that $\zeta _{i} (\boldsymbol{p} ,r) \geq  -\xi $ for all $1 \leq i \leq n$ and all $(\boldsymbol{p} ,r) \in \boldsymbol{P}$, and $\zeta _{n+1} (\boldsymbol{p} ,r) \geq  -\xi $ for all $(\boldsymbol{p} ,r) \in \boldsymbol{P}$ such that $r \leq \delta < 1$ for some $\delta \in (0,1)$. 
\end{lemma}

\begin{proof}
We have $\zeta _{i} (\boldsymbol{p} ,r) \geq  -\sum _{y_{i} \in \mathcal{Y}_{i}}^{\,}e_{i} (y_{i}) y_{i}$ for all $1 \leq i \leq n$ and all $(\boldsymbol{p} ,r) \in \boldsymbol{P}$. Thus, $\zeta _{i}$ is bounded from below for all $1 \leq i \leq n$. 

Since $g\left (a ,\boldsymbol{p} ,r\right )$ is bounded from above by $\sum \limits_{i =1}^{n}p_{i} \overline{b}/ (1-r)^{2}$ 
we have 
$$\int g\left (a ,\boldsymbol{p} ,r\right )\mu \left (da;\boldsymbol{p} ,r\right ) \leq \sum \limits_{i =1}^{n}p_{i} \overline{b}/(1-r)^{2},$$ so
\begin{equation*}\zeta _{n +1}\left (\boldsymbol{p} ,r\right ) = -\int g\left (a ,\boldsymbol{p} ,r\right )\mu \left (da;\boldsymbol{p} ,r\right ) \geq  -\frac{ \sum \limits_{i =1}^{n}p_{i} \overline{b} }{(1-r)^{2}} \geq - \frac {\overline{b} } {(1-\delta)^{2} }
\end{equation*} for all $(\boldsymbol{p} ,r) \in \boldsymbol{P}$ such that $r \leq \delta <1$.
\end{proof}

\begin{lemma}
\label{Lemma boundary} (i) If  $(\boldsymbol{p}_{q} ,r_{q}) \rightarrow (\boldsymbol{p} ,r) =(p_{1}, \ldots ,p_{n} ,r) \in \Lambda \backslash \boldsymbol{P}$ with $\{\boldsymbol{p}_{q} ,r_{q}\} \subseteq \boldsymbol{P}$ and $p_{k}>0$ for some $1\leq k \leq n$, then $\lim _{q \rightarrow \infty }\left \Vert \zeta (\boldsymbol{p}_{q} ,r_{q})\right \Vert _{1} =\infty $. 

(ii) If $\{\boldsymbol{p}_{q} ,r_{q}\} \subseteq \boldsymbol{P}$, $(\boldsymbol{p}_{q} ,r_{q}) \rightarrow (\boldsymbol{p} ,r) =(p_{1} ,\ldots  ,p_{n} ,r)$ and $p_{k} >0$, then the sequence $\{\zeta _{k} (\boldsymbol{p}_{q} ,r_{q})\}$ of the $k^{t h}$ components of $\{\zeta  (\boldsymbol{p}_{q} ,r_{q})\}$ is bounded. Similarly, $r \in (0,1)$ implies that the sequence $\{\zeta _{n +1} (\boldsymbol{p}_{q} ,r_{q})\}$ is bounded. 
\end{lemma}

\begin{proof}
(i) Suppose that $(\boldsymbol{p}_{q} ,r_{q}) \rightarrow (\boldsymbol{p} ,r) =(p_{1} ,\ldots  ,p_{n} ,r)$ where $(\boldsymbol{p} ,r) \in \Lambda \backslash \boldsymbol{P}$ and $p_{k} >0$. We consider two different cases. 

Case I: We have $r_{q} \rightarrow r =0$. In this case the borrowing constraint tends to minus infinity and it follows from the same arguments as the arguments on page 673 in \cite{aiyagari1994}
that 
\begin{equation*}
\underset{q \rightarrow \infty }{\lim }\int g (a,\boldsymbol{p}_{q},r_{q}) \mu  (d a ;\boldsymbol{p}_{q} ,r_{q}) = -\infty.
\end{equation*} 
Thus, we have $\underset{q \rightarrow \infty }{\lim }\zeta _{n +1} (\boldsymbol{p}_{q} ,r_{q}) =\infty $ which implies that $\underset{q \rightarrow \infty }{\lim }\left \Vert \zeta  (\boldsymbol{p}_{q} ,r_{q})\right \Vert _{1} =\infty $. 

Case II: We have $0 < r < 1$. In this case $(\boldsymbol{p},r) \in \Lambda \backslash \boldsymbol{P}$ implies that $p_{j} =0$ and $p_{k}>0$ for some $1 \leq k \leq n$ and $1 \leq j \leq n$. 

The fact that $p_{k} > 0$ (and hence, $\boldsymbol{p}_{q} \cdot \boldsymbol{y}$ is bounded away from $0$) implies that the  measurable set $D_{\epsilon} (\boldsymbol{p}_{q},r_{q}) \subseteq A(\boldsymbol{p}_{q},r_{q}) $ given by  $$D_{\epsilon} (\boldsymbol{p}_{q},r_{q}) := \{a \in A(\boldsymbol{p}_{q},r_{q}):  a-g(a,\boldsymbol{p}_{q},r_{q}) \geq \epsilon \}$$   
satisfies $\inf _{q \geq N} \mu (D_{\epsilon} ;\boldsymbol{p}_{q},r_{q}) \geq \delta ' $ 
for some $\epsilon > 0$, $\delta ' > 0$ and $N>0$. 
To see this, assume in contradiction that this is not true. Then for all $\epsilon > 0 $, $\delta ' > 0$,  and $N>0$ we can find a $q \geq N$ such that $\mu (D_{\epsilon} ;\boldsymbol{p}_{q},r_{q}) < \delta '$.  

Let $\boldsymbol{y}' \in \mathcal{Y}$ satisfies $\boldsymbol{y}' > \boldsymbol{y}$ for all $\boldsymbol{y} \in \mathcal{Y}$. Then because $p_{k} >0$ we can choose a large $N>0$ such that $e(\boldsymbol{y}')(\boldsymbol{p}_{q} \cdot \boldsymbol{y}'- \min_{\boldsymbol{y} \in \mathcal{Y}} \boldsymbol{p}_{q} \cdot \boldsymbol{y} )\geq \delta$ for some $\delta >0$ that does not depend on $q$ and for all $q \geq N$.  
Let $M \in \mathbb{R}$ be such that $M \geq a-g (a ,\boldsymbol{p}_{q} ,r_{q})$ for all $q \geq N$ (because $r \in (0,1)$ we have $A(\boldsymbol{p}_{q} ,r_{q}) \subseteq A$ for some  compact set $A$ and all $q \geq N$, so $M$ exists). 

We can choose $\epsilon$ and $\delta '$ such that $0<\epsilon + \delta ' M < \delta$. 

Let  $q \geq N$ and assume in contradiction that $\mu (D_{\epsilon} ;\boldsymbol{p}_{q},r_{q}) < \delta ' $.   Note that 

\begin{align*}  \int _{A(\boldsymbol{p}_{q},r_{q})}(a -g (a ,\boldsymbol{p}_{q} ,r_{q})) \mu  (d a ;\boldsymbol{p}_{q} ,r_{q})  & =   \int _{D_{\epsilon} (\boldsymbol{p}_{q},r_{q})}(a -g (a ,\boldsymbol{p}_{q} ,r_{q})) \mu  (d a ;\boldsymbol{p}_{q} ,r_{q})  \\
& + \int _{A(\boldsymbol{p}_{q},r_{q}) \setminus D_{\epsilon} (\boldsymbol{p}_{q},r_{q})}(a -g (a ,\boldsymbol{p}_{q} ,r_{q})) \mu  (d a ;\boldsymbol{p}_{q} ,r_{q}) \\
& \leq M \delta ' + \epsilon < \delta.  
\end{align*} 
 On the other hand, we have
\begin{align*} \int _{A(\boldsymbol{p}_{q},r_{q})}(a -g (a ,\boldsymbol{p}_{q} ,r_{q})) \mu  (d a ;\boldsymbol{p}_{q} ,r_{q}) 
& = \int _{A(\boldsymbol{p}_{q},r_{q})}r_{q}g\left (a ,\boldsymbol{p}_{q} ,r_{q}\right )\mu \left (da;\boldsymbol{p}_{q} ,r_{q}\right ) +\sum _{\boldsymbol{y} \in \mathcal{Y}}e\left (\boldsymbol{y}\right )\boldsymbol{p}_{q} \cdot \boldsymbol{y} \\
 & \geq - \int _{A(\boldsymbol{p}_{q},r_{q})}\min_{\boldsymbol{y} \in \mathcal{Y}} \boldsymbol{p}_{q} \cdot \boldsymbol{y} \ \mu \left (da;\boldsymbol{p}_{q} ,r_{q}\right ) +  \sum _{\boldsymbol{y} \in \mathcal{Y}}e\left (\boldsymbol{y}\right )\boldsymbol{p}_{q} \cdot \boldsymbol{y} \\ 
 & =  -\min_{\boldsymbol{y} \in \mathcal{Y}} \boldsymbol{p}_{q} \cdot \boldsymbol{y} + \sum _{\boldsymbol{y} \in \mathcal{Y}}e\left (\boldsymbol{y}\right )\boldsymbol{p}_{q} \cdot \boldsymbol{y} \\ 
 & = -\min_{\boldsymbol{y} \in \mathcal{Y}} \boldsymbol{p}_{q} \cdot \boldsymbol{y} + e(\boldsymbol{y}')\boldsymbol{p}_{q} \cdot \boldsymbol{y}'+\sum _{\boldsymbol{y} \in \mathcal{Y} \setminus \{\boldsymbol{y}' \} }e\left (\boldsymbol{y}\right )\boldsymbol{p}_{q} \cdot \boldsymbol{y} \\ 
 & \geq -\min_{\boldsymbol{y} \in \mathcal{Y}} \boldsymbol{p}_{q} \cdot \boldsymbol{y} + e(\boldsymbol{y}')\boldsymbol{p}_{q} \cdot \boldsymbol{y}'+(1-e(\boldsymbol{y}')) \min_{\boldsymbol{y} \in \mathcal{Y}} \boldsymbol{p}_{q} \cdot \boldsymbol{y} \\
 & = e(\boldsymbol{y}')(\boldsymbol{p}_{q} \cdot \boldsymbol{y}'- \min_{\boldsymbol{y} \in \mathcal{Y}} \boldsymbol{p}_{q} \cdot \boldsymbol{y} )\geq \delta > 0\end{align*}
which is a contradiction. 
The first equality follows from Equation (\ref{eq:stationary}). The second inequality follows from the borrowing constraint.

Now assume in contradiction that  $\{\boldsymbol{x}^{\ast} (a-g(a,\boldsymbol{p_{q}},r_{q}),\boldsymbol{p_{q}}) \}$ has a bounded subsequence.

Let $a \in D_{\epsilon} (\boldsymbol{p}_{q},r_{q}) $ where $\inf _{q \geq N} \mu (D_{\epsilon} ;\boldsymbol{p}_{q},r_{q}) \geq \delta $ 
for some $\epsilon > 0$, $\delta>0$, and  $N>0$. 
Because   $\{\boldsymbol{x}^{\ast} (a-g(a,\boldsymbol{p_{q}},r_{q}),\boldsymbol{p_{q}}) \}$ has a bounded subsequence, then by passing to a subsequence (and relabelling if needed), we can assume that $\boldsymbol{x}^{\ast} (a-g(a,\boldsymbol{p_{q}},r_{q}),\boldsymbol{p_{q})} \rightarrow \boldsymbol{x}$ holds in $\mathbb{R}^{n}_{+}$. 

With slight abuse of notation we extend the savings policy function and let $g(a,\boldsymbol{p},r) :=  \limsup g(a,\boldsymbol{p}_{q},r_{q}) $. 
Because $r \in (0,1)$, the function $g(a,\boldsymbol{p},r)$ is finite.

We claim that $\boldsymbol{x} = \boldsymbol{x}^{\ast}(a-g(a,\boldsymbol{p},r),\boldsymbol{p})$. 
The fact that $a - g(a,\boldsymbol{p}_{q},r_{q}) \geq \epsilon$ implies that $a - g(a,\boldsymbol{p},r) \geq \epsilon$. Furthermore, 
$$\boldsymbol{p} \cdot \boldsymbol{x} =  \lim_{q \rightarrow \infty} \boldsymbol{p}_{q} \cdot \boldsymbol{x}^{\ast} (a-g(a,\boldsymbol{p_{q}},r_{q}),\boldsymbol{p_{q})} 
= \limsup_{q \rightarrow \infty}  (a- g(a,\boldsymbol{p}_{q},r_{q})) = a - g(a,\boldsymbol{p},r) $$
which implies that $\boldsymbol{x} \in X(a-g(a,\boldsymbol{p},r),\boldsymbol{p})$ (and that  $\lim_{q \rightarrow \infty}  (a- g(a,\boldsymbol{p}_{q},r_{q}))$ exists). Now let $\boldsymbol{x}' \in X(a-g(a,\boldsymbol{p},r),\boldsymbol{p}) $. Because  $a - g(a,\boldsymbol{p},r) > 0$ for every $\lambda \in (0,1)$ we have $ \boldsymbol{p} \cdot (\lambda \boldsymbol{x}' ) < a-g(a,\boldsymbol{p},r)$. Hence, there exists a number $N_{0}$ such that for every $q > N_{0}$ we have 
$$\boldsymbol{p}_{q} \cdot (\lambda \boldsymbol{x}' )  < a - g(a,\boldsymbol{p}_{q},r_{q}) = \boldsymbol{p}_{q} \cdot  \boldsymbol{x}^{\ast} (a- g(a,\boldsymbol{p}_{q},r_{q}),  \boldsymbol{p}_{q}).$$  
Using the monotonicity of the utility function we conclude that the consumption bundle $\boldsymbol{x}^{\ast} (a- g(a,\boldsymbol{p}_{q},r_{q}),  \boldsymbol{p}_{q})$ is preferred to the consumption bundle  $\lambda \boldsymbol{x}'$ for all $\lambda \in (0,1)$. Continuity of $U$ implies that  $U (\boldsymbol{x}) \geq U (\boldsymbol{x}')$. Thus, $\boldsymbol{x} = \boldsymbol{x}^{\ast}(a-g(a,\boldsymbol{p},r),\boldsymbol{p})$, i.e., $\boldsymbol{x}$ maximizes $U$ on $X(a-g(a,\boldsymbol{p},r),\boldsymbol{p})$

Let $\boldsymbol{x}''$ be a consumption bundle such that $x_{i}'' = x_{i}$ for all $i \neq j$ and $x_{j}'' = x_{j} + 1$. Then $\boldsymbol{x}''$ is feasible in $X(a-g(a,\boldsymbol{p},r),\boldsymbol{p})$  (recall that $p_{j}=0$). From the strict monotonicity of $U$, $\boldsymbol{x}''$ yields strictly more utility than $\boldsymbol{x}$. This contradicts that $\boldsymbol{x}$ maximizes $U$ on $X(a-g(a,\boldsymbol{p},r),\boldsymbol{p})$. 

We conclude that the sequence $\{\boldsymbol{x}^{\ast} (a-g(a,\boldsymbol{p_{q}},r_{q}),\boldsymbol{p_{q}}) \}$ diverges. Hence,  
 $$\lim _{q \rightarrow \infty }\sum _{i =1}^{n}x_{i}^{\ast} \left (a -g\left (a ,\boldsymbol{p}_{q} ,r_{q}\right ) ,\boldsymbol{p}_{q}\right ) =\infty $$ for all $a \in D_{\epsilon} (\boldsymbol{p}_{q},r_{q}) := D_{\epsilon}^{q}$. We have 
\begin{align*} 
\lim _{q \rightarrow \infty} \int _{A(\boldsymbol{p}_{q},r_{q})} \sum _{i =1}^{n} x_{i}^{ \ast }(a -g(a ,\boldsymbol{p}_{q} ,r_{q}) ,\boldsymbol{p}_{q})\mu (da;\boldsymbol{p}_{q} ,r_{q}) & 
=  \lim _{q \rightarrow \infty} \int _{ D_{\epsilon}^{q} } \sum _{i =1}^{n} x_{i}^{ \ast }(a -g(a ,\boldsymbol{p}_{q} ,r_{q}) ,\boldsymbol{p}_{q})\mu (da;\boldsymbol{p}_{q} ,r_{q}) \\
& +\lim _{q \rightarrow \infty}  \int _{A(\boldsymbol{p}_{q},r_{q}) \setminus D_{\epsilon}^{q} } \sum _{i =1}^{n}x_{i}^{ \ast }(a -g(a ,\boldsymbol{p}_{q} ,r_{q}) ,\boldsymbol{p}_{q})\mu (da;\boldsymbol{p}_{q} ,r_{q}) \\ 
& = \infty
\end{align*}
i.e., 
 $\underset{q \rightarrow \infty }{\lim }\left \Vert \zeta  (\boldsymbol{p}_{q} ,r_{q})\right \Vert _{1} =\infty $.

(ii) Assume that $\{\boldsymbol{p}_{q} ,r_{q}\}$ is a sequence of strictly positive prices satisfying the conditions of the
Lemma where $\boldsymbol{p}_{q} =(p_{1}^{q} ,\ldots  ,p_{n}^{q})$. Since $p_{k} >0$ for some $1 \leq k \leq n$ and $(\boldsymbol{p}_{q} ,r_{q}) \rightarrow (\boldsymbol{p} ,r)$, we can assume that there exists some $\epsilon  >0$ such that $p_{k}^{q} >\epsilon $ for all $q$. Let $M$ be such that $\sum _{\boldsymbol{y} \in \mathcal{Y}}e\left (\boldsymbol{y}\right )\boldsymbol{p}_{q} \cdot \boldsymbol{y} \leq M$ for all $q$.

From the budget constraint we have 
\begin{equation*}p_{k}^{q} \int _{A(\boldsymbol{p}_{q},r_{q})}x_{k}^{ \ast } (a -g (a ,\boldsymbol{p}_{q} ,r_{q}) ,\boldsymbol{p}_{q}) \mu  (d a ;\boldsymbol{p}_{q} ,r_{q}) \leq \int _{A(\boldsymbol{p}_{q},r_{q})}(a -g (a ,\boldsymbol{p}_{q} ,r_{q})) \mu  (d a ;\boldsymbol{p}_{q} ,r_{q})\text{.}
\end{equation*}
The last inequality implies that
\begin{align*}\int _{A(\boldsymbol{p}_{q},r_{q})}x_{k}^{ \ast } (a -g (a ,\boldsymbol{p}_{q} ,r_{q}) ,\boldsymbol{p}_{q}) \mu  (d a ;\boldsymbol{p}_{q} ,r_{q}) &  \leq \frac{\int _{A(\boldsymbol{p}_{q},r_{q})}(a -g (a ,\boldsymbol{p}_{q} ,r_{q})) \mu  (d a ;\boldsymbol{p}_{q} ,r_{q})}{p_{k}^{q}} \\
 &  =\frac{\int _{A(\boldsymbol{p}_{q},r_{q})}r_{q}g\left (a ,\boldsymbol{p}_{q} ,r_{q}\right )\mu \left (da;\boldsymbol{p}_{q} ,r_{q}\right ) +\sum _{\boldsymbol{y} \in \mathcal{Y}}e\left (\boldsymbol{y}\right )\boldsymbol{p}_{q} \cdot \boldsymbol{y}}{p_{k}^{q}} \\
 &  \leq \frac{r_{q}\sum _{i =1}^{n}p_{i}^{q}\overline{b}/(1-r_{q})^{2} +\sum _{\boldsymbol{y} \in \mathcal{Y}}e\left (\boldsymbol{y}\right )\boldsymbol{p}_{q} \cdot \boldsymbol{y}}{p_{k}^{q}} \leq \frac{\overline{b}/(1-\epsilon)^{2} +M}{\epsilon } .
 \end{align*} 
The equality follows from Equation (\ref{eq:stationary}). The second inequality follows since $g\left (a ,\boldsymbol{p}_{q} ,r_{q}\right )$$ \leq \sum p_{i}^{q}\overline{b} / (1-r_{q})^{2}$ for all $a \in A(\boldsymbol{p}_{q},r_{q})$. Therefore, the sequence $\{\zeta _{k} (\boldsymbol{p}_{q} ,r_{q})\}$ is bounded for $1 \leq k \leq n$.

Now assume that $r \in (0,1)$. In this case, we can assume that there exists $\delta  >0$ such that $r_{q} >\delta $ for all $q$. We can also assume that $\min _{\boldsymbol{y} \in \mathcal{Y}}\boldsymbol{p}_{q} \cdot \boldsymbol{y} \leq \overline{M}$ for all $q$.  

Using the borrowing constraint, we have \begin{equation*} -\int _{A(\boldsymbol{p}_{q},r_{q})}g\left (a ,\boldsymbol{p}_{q} ,r_{q}\right )\mu \left (da;\boldsymbol{p}_{q} ,r_{q}\right ) \leq \frac{\min _{\boldsymbol{y} \in \mathcal{Y}}\boldsymbol{p}_{q}\boldsymbol{y}}{r_{q}} \leq \frac{\overline{M}}{\delta } .
\end{equation*} Therefore the sequence $\{\zeta _{n +1} (\boldsymbol{p}_{q} ,r_{q})\}$ is bounded from above. From Lemma \ref{Lemma Bdd from below}, $\{\zeta _{n +1} (\boldsymbol{p}_{q} ,r_{q})\}$ is bounded from below. The proof of the Lemma is completed.
\end{proof}

 We now prove that a CSE exists. The proof follows similar arguments to the proof of  Theorem 1.4.8 in \cite{aliprantis1990existence}). However, as discussed in Remark \ref{Remark:Boundary}, we cannot use this result directly because the standard boundary and boundness conditions do not hold in our setting.

Let $\zeta ( \cdot ) =(\zeta _{1}( \cdot ) , \ldots ,\zeta _{n +1}( \cdot ))$ from $\boldsymbol{P}$ into $\mathbb{R}^{n +1}$ be a function that satisfies the properties of Proposition \ref{Prop: existence}. For notation simplicity, we will write $p_{n+1} := r$ and $\boldsymbol{p} = (p_{1},\ldots,p_{n+1})$ for the rest of the proof. 
Define 
$$ W (\boldsymbol{p}) =   \{ k \in  \{1,\ldots,n+1 \} : \ \zeta_{k}(\boldsymbol{p}) = \max \{ \zeta_{i}(\boldsymbol{p}):\ i=1,\ldots, n+1 \} \  \}$$
for each $\boldsymbol{p} \in \boldsymbol{P}$ and 
$$ W (\boldsymbol{p}) =   \{ k \in  \{1,\ldots,n+1 \} : \ p_{k} = 0  \  \}$$
for each $\boldsymbol{p} \in \Lambda \setminus \boldsymbol{P}$. 

We also define the correspondence $ \Phi : \Lambda \rightarrow 2^{\Lambda}$ where $2^{\Lambda}$ is the set of all subsets of $\Lambda$
$$ \Phi (\boldsymbol{p}) =  \{ z \in \Lambda: \ z_{k} = 0 \text{ for all } k \notin W (\boldsymbol{p}) \} .$$ 
It is immediate to establish that $\Phi (\boldsymbol{p})$ is a non-empty, convex and compact subset of $\Lambda$. We claim that $\Phi$ has a closed graph.\footnote{Recall that $\Phi$ has a closed graph whenever the graph of $\Phi$ $$ G_{\phi} = \{(\boldsymbol{p},\boldsymbol{p}') \in \Lambda \times \Lambda: \boldsymbol{p} \in \Lambda \ \boldsymbol{p}' \in \Phi(\boldsymbol{p}) \}$$ 
is a closed subset of $\Lambda \times \Lambda$.} Assume that $\boldsymbol{p}^{q} \rightarrow \boldsymbol{p}$ in  $\Lambda$, $\boldsymbol{\pi } ^{q} \rightarrow \boldsymbol { \pi }$ in $\Lambda$, and $ \boldsymbol { \pi} ^{q} \in  \Phi (\boldsymbol{p}^{q})$ for each $q$. We have to show that $\boldsymbol {\pi} \in \Phi(\boldsymbol{p})$. 

The case that $\boldsymbol{p} \in \boldsymbol{P}$ is easy to establish and does not depend on the boundary and boundness conditions (see  the proof of Theorem 1.4.8 in \cite{aliprantis1990existence}). 

Let $\boldsymbol{p} \in \Lambda \setminus \boldsymbol{P}$. 
We consider two cases. 

(i) There exists a subsequence of $\{ \boldsymbol{p}^{q} \}$ (by relabelling if needed, we can assume it to be $\{\boldsymbol{p}^{q} \}$) in $\boldsymbol{P}$.

First assume that $p_{n+1} = 1$ (i.e., $r=1$). Then $p^{q}_{n+1} = r^{q} > 1/\beta - 1 $ for all $q > N$ and some $N>0$. Combining the results in \cite{chamberlain2000optimal} and  \cite{accikgoz2018existence} that show that the agents' reach the maximal wealth they can have with probability 1 for any initial wealth level and the fact that the upper bound on savings tends to infinity\footnote{Note that for $(\boldsymbol{p},r) \in \boldsymbol{P}$ we have $\sum _{i=1}^{n} p_{i} / (1-r)^{2} = 1/\sum _{i=1}^{n} p_{i}$ so the upper bound on savings tends to infinity when $\sum _{i=1}^{n} p_{i}$ tends to $0$.} when $p^{q}_{i} \rightarrow 0$ for all $i=1,\ldots,n$  imply that $\zeta _{n+1}  (\boldsymbol{p}^{q}) < 0$ for all $q > N$. Walrals' law implies that  $\zeta _{i} (\boldsymbol{p}^{q}) > 0$ for some $1 \leq i \leq n$. Therefore, $ W (\boldsymbol{p}^{q}) \subseteq \{1,\ldots,n \} = W (\boldsymbol{p})$ for all $q > N$. From $ \boldsymbol  { \pi}  ^{q} \in  \Phi (\boldsymbol{p}^{q}) $ it follows that $\pi _{k}^{q} =0$ for all $k \notin W (\boldsymbol{p}^{q})$  and for all $q > N$. Hence, $\lim_{q \rightarrow \infty} \pi _{k}^{q} = \pi _{k} = 0$ for all $k \notin W (\boldsymbol{p})$. We conclude that $\boldsymbol {\pi}  \in  \Phi (\boldsymbol{p})$.

Now assume that $p_{n+1} \leq \delta < 1$ for some $\delta <1$. We can assume that $p^{q}_{n+1} \leq \delta<1$ for all $q$. We have $p_{k} > 0$ for some $1 \leq k \leq n$. From Proposition \ref{Prop: existence} part (iv) the sequence $\{ \zeta_{i} ( \boldsymbol{p}^{q} ) \}$ is bounded for each $i \notin  W (\boldsymbol{p})$. From Proposition \ref{Prop: existence} part (iii) we have $\lim _{q \rightarrow \infty }\left \Vert \zeta (\boldsymbol{p}^{q})\right \Vert _{1} =\infty $. Because $\zeta_{i}(\boldsymbol{p}^{q})$ is bounded from below for $i=1,\ldots,n+1$ when  $p^{q}_{n+1} \leq \delta<1$, we conclude that $\max \{ \zeta_{i} (\boldsymbol{p}^{q}) : i=1,\ldots ,n+1 \}$ is an element of $ W (\boldsymbol{p})$  for all $q > N$ and some $N>0$. Thus, there exists some $N$ such that $ W (\boldsymbol{p}^{q}) \subseteq W (\boldsymbol{p})$ for all $q > N$.
Combining this and the fact that $\boldsymbol{ \pi ^{q} } \in  \Phi (\boldsymbol{p}^{q})$ implies $\boldsymbol{ \pi } ^{q} \in  \Phi (\boldsymbol{p})$ for all $q > N$. Hence, $\boldsymbol{ \pi }  \in  \Phi (\boldsymbol{p})$. 

(ii) No subsequence of $\{ \boldsymbol{p}^{q} \} $ lies in $\boldsymbol{P}$. Then $p^{q}_{i} \rightarrow p_{i}$ for all $i$ immediately implies that $W (\boldsymbol{p}^{q}) \subseteq W (\boldsymbol{p})$. From the argument above  we have $\boldsymbol{ \pi}  \in  \Phi (\boldsymbol{p})$.  

Thus, we have established that $\phi$ has a closed graph. The rest of the proof follows the proof of 
Theorem 1.4.8 in \cite{aliprantis1990existence}). We provide it here for completeness. 

From Kakutani's fixed point theorem (see Theorem 1.4.7 in \cite{aliprantis1990existence}), $\phi$ has a fixed point, say $\boldsymbol{p}$, i.e., $\boldsymbol{p} \in \phi(\boldsymbol{p})$. We claim that $\boldsymbol{p}$ is a CSE. First note that $\boldsymbol{p} \gg 0 $. To see this assume in contradiction that $p_{k} = 0 $ for some $k$. Then, $p_{j} = 0$ for all $j \in W(\boldsymbol{p}) $. Furthermore, $\boldsymbol{p} \in \phi(\boldsymbol{p})$ implies that $p_{j} = 0$ for all $j \notin W(\boldsymbol{p}) $. We conclude that $p_{k}=0$ for all $k=1,\ldots,n+1$ which contradict the fact that  $\boldsymbol{p} \in \Lambda$. Hence, $\boldsymbol{p} \gg 0 $.

Now, let 
$ \gamma  = \max \{ \zeta_{i}(\boldsymbol{p},r):\ i=1,\ldots, n+1 \}$. Note that $\boldsymbol{p} \in \phi(\boldsymbol{p})$ and $\boldsymbol{p} \gg 0 $ imply $i \in W(\boldsymbol{p})$ for $i=1,\ldots,n+1$. Hence, $\gamma = \zeta_{i} (\boldsymbol{p}) $ for $i=1,\ldots,n+1$. From Warlas' law we have 
$$\gamma = \sum_{i=1}^{n+1} p_{i} m = \sum _{i=1}^{n+1} p_{i} \zeta_{i} (\boldsymbol{p}) = 0  $$
which implies $\zeta (\boldsymbol{p}) = 0$ because $\boldsymbol{p} \gg 0 $. Thus, $\boldsymbol{p}$ is a CSE.

\subsection{The uniqueness of a competitive stationary equilibrium}
In this section we prove Theorem \ref{Theorem uniq}. 

\noindent \textbf{Theorem \ref{Theorem uniq}.}
\emph{Assume that $U(\boldsymbol{x}) =\sum _{i =1}^{n}\alpha _{i}x_{i}^{\gamma }$ for some $0 <\gamma  <1$, $\alpha _{i} >0$, $\sum _{i =1}^{n}\alpha _{i} =1$. Then there exists a unique competitive stationary equilibrium.}

\begin{proof} Since the savings policy function $g$, the demand function $\boldsymbol{x}^{ \ast }$, and the invariant wealth distribution $\mu $ are unique given fixed prices $(\boldsymbol{p} ,r)$, it is enough to show that the prices $(\boldsymbol{p},r)$ that clear the market are unique in order to prove the uniqueness of a CSE. The proof involves a number of steps. 

\textbf{Step 1.} If $(\boldsymbol{p} ,r)$ and $(\boldsymbol{p}^{ \prime } ,r^{ \prime })$ are equilibrium prices, then $\boldsymbol{p} =\boldsymbol{p}^{ \prime }$. Suppose, in contradiction, that there are equilibrium prices $(\boldsymbol{p} ,r)$ and $(\boldsymbol{p}^{ \prime } ,r^{ \prime })$ such that $\boldsymbol{p}^{ \prime } \neq \boldsymbol{p}$, and $\boldsymbol{p}$ and $\boldsymbol{p}^{ \prime }$ are not linearly independent. From Proposition \ref{Prop homogenous of }, we can normalize the prices such that $\boldsymbol{p} \geq \boldsymbol{p}^{ \prime }$ and $p_{k}^{ \prime } =p_{k} =1$ for some $1 \leq k \leq n$. We have
\begin{align*}\int a\mu (da;\boldsymbol{p} ,r) =\sum \limits _{\boldsymbol{y} \in \mathcal{Y}}e(\boldsymbol{y})((1 +r)\int g(a ,\boldsymbol{p} ,r)\mu (da;\boldsymbol{p} ,r) +\boldsymbol{p} \cdot \boldsymbol{y}) =\sum \limits _{\boldsymbol{y} \in \mathcal{Y}}^{\,}e(\boldsymbol{y})\boldsymbol{p} \cdot \boldsymbol{y}.\end{align*}
The first equality follows from Equation (\ref{eq:stationary}). The second equality follows from the fact that $(\boldsymbol{p} ,r)$ are equilibrium prices. 

Similarly,  $\int a\mu (da;\boldsymbol{p}^{ \prime } ,r^{ \prime }) =\sum \limits _{\boldsymbol{y} \in \mathcal{Y}}^{\,}e(\boldsymbol{y})\boldsymbol{p}^{ \prime } \cdot \boldsymbol{y}$. Since $\boldsymbol{y} \gg 0$ we conclude that $\int a\mu (da;\boldsymbol{p}^{ \prime } ,r^{ \prime }) <\int a\mu (da;\boldsymbol{p} ,r)$. 

Using the fact that $\int g(a ,\boldsymbol{p}^{ \prime } ,r^{ \prime })\mu (da;\boldsymbol{p}^{ \prime } ,r^{ \prime }) = \int g(a ,\boldsymbol{p} ,r)\mu(da;\boldsymbol{p},r)=0$ we have \begin{equation*}\int (a -g(a ,\boldsymbol{p}^{ \prime } ,r^{ \prime }))\mu (da;\boldsymbol{p}^{ \prime } ,r^{ \prime }) <\int (a -g(a ,\boldsymbol{p} ,r))\mu (da;\boldsymbol{p} ,r) .
\end{equation*}
Since the utility function is in the constant elasticity of substitution class, it is well known and easy to check that $\boldsymbol{x}^{ \ast }(a -g(a ,\boldsymbol{p} ,r) ,\boldsymbol{p}) =(z_{1}(\boldsymbol{p})(a -g(a ,\boldsymbol{p} ,r)) , . . . ,z_{n}(\boldsymbol{p})(a -g(a ,\boldsymbol{p} ,r))$ where $z_{i}(\boldsymbol{p})$ is a positive function for each $i =1 ,\ldots  ,n$. Thus, $x_{i}^{ \ast }$ is linear in the total expenditure $a -g(a ,\boldsymbol{p} ,r)$ for all $i$. From the assumption that the elasticity of substitution is higher than one, the demand for each good increases with the prices of the other goods. Since $\boldsymbol{p} \geq \boldsymbol{p}^{ \prime }$ and $p_{k}^{ \prime } =p_{k} =1$ we have
 $z_{k}(\boldsymbol{p}) \geq z_{k}(\boldsymbol{p}^{ \prime })$.  
 
 We have 
\begin{align*}\int x_{k}^{ \ast }(a -g(a ,\boldsymbol{p} ,r) ,\boldsymbol{p})\mu (da;\boldsymbol{p} ,r) &  =\int z_{k}(\boldsymbol{p})(a -g(a ,\boldsymbol{p} ,r))\mu (da;\boldsymbol{p} ,r) \\
 &  >\int z_{k}(\boldsymbol{p}^{ \prime })(a -g(a ,\boldsymbol{p}^{ \prime } ,r^{ \prime }))\mu (da;\boldsymbol{p}^{ \prime } ,r^{ \prime }) \\
 &  =\int x_{k}^{ \ast }(a -g(a ,\boldsymbol{p}^{ \prime } ,r^{ \prime }) ,\boldsymbol{p}^{ \prime })\mu (da;\boldsymbol{p}^{ \prime } ,r^{ \prime }) ,\end{align*}which leads to the contradiction $0 =\zeta _{k}(\boldsymbol{p} ,r) >\zeta _{k}(\boldsymbol{p}^{ \prime } ,r^{ \prime }) =0$.

\textbf{Step 2. }$g(a ,\boldsymbol{p} ,r)$ is increasing and convex in $a$ for all $(\boldsymbol{p} ,r)\in\boldsymbol{P}$. It is easy to check that the indirect utility function $v(a -b ,\boldsymbol{p}) =\max _{x \in X(a -b ,\boldsymbol{p})}U(\boldsymbol{x})$ is given by $v(a -b ,\boldsymbol{p}) =(a -b)^{\gamma }z(\boldsymbol{p})$ where $z(\boldsymbol{p})$ is a positive function. The indirect utility function is a constant relative risk aversion utility function and thus the savings policy function is convex in $a$ (for example, we can apply Theorem 4 in \cite{jensen2017distributional} or the results in \cite{huggett2004precautionary}). 

To show that $g$ is increasing in $a$, note that $v(a -b ,\boldsymbol{p})$ has increasing differences in $(a ,b)$ (recall that a function $v$ is said to have increasing differences in $(a ,b)$ if for all $a_{2} \geq a_{1}$ and $b_{2} \geq b_{1}$ we have $v(a_{2} -b_{2} ,\boldsymbol{p}) -v(a_{2} - b_{1} ,\boldsymbol{p}) \geq v(a_{1} - b_{2} ,\boldsymbol{p}) -v(a_{1} - b_{1} ,\boldsymbol{p})$). Thus, the function \begin{equation*}v(a -b ,\boldsymbol{p}) +\beta \sum \limits _{\boldsymbol{y} \in \mathcal{Y}}^{\,}e(\boldsymbol{y})V((1 +r)b +\boldsymbol{p} \cdot \boldsymbol{y} ,\boldsymbol{p} ,r)
\end{equation*} has increasing differences in $(a ,b)$ as the sum of functions with increasing differences. Now Theorem 6.1 in \cite{topkis1978minimizing} implies that $g(a ,\boldsymbol{p} ,r)$ is increasing in $a$.

\textbf{Step 3.} $g(a ,\boldsymbol{p} ,r)$ is increasing in $r$ for all $(a ,\boldsymbol{p})$. The proof of this result follows from similar arguments to the arguments in the proof of Theorem 1 in \cite{light2017uniqueness}. Since the current setting is different from the setting in  \cite{light2017uniqueness} we provide the proof here.

Assume that $f(a ,\boldsymbol{p} ,r)$ is a bounded function that is increasing, concave and continuously differentiable in $a$ with the following properties: (i) $f$ has increasing differences in $(a ,r)$; (ii) $af_{a}(a ,\boldsymbol{p} ,r)\ $is increasing in $a$ on $\mathbb{R}_{+}$ (for a function $f$ we denote by $f_{a}$ the derivative of $f$ with respect to $a$). Let  $1>r >r^{ \prime }>0$. We have  
\begin{align*}(1 +r)f_{a}((1 +r)b +\boldsymbol{p} \cdot \boldsymbol{y} ,\boldsymbol{p} ,r) &  \geq (1 +r^{ \prime })f_{a}((1 +r^{ \prime })b +\boldsymbol{p} \cdot \boldsymbol{y} ,\boldsymbol{p} ,r) \\
 &  \geq (1 +r^{ \prime })f_{a}((1 +r^{ \prime })b +\boldsymbol{p} \cdot \boldsymbol{y} ,\boldsymbol{p} ,r^{ \prime }) .\end{align*}The first inequality follows from property (ii) if\protect\footnote{
To see this, let $a =(1 +z)b +\boldsymbol{p} \cdot \boldsymbol{y}$. Then $af_{a}(a ,\boldsymbol{p} ,r) =b(1 +z)f_{a}((1 +z)b +\boldsymbol{p} \cdot \boldsymbol{y} ,\boldsymbol{p} ,r) +\boldsymbol{p} \cdot \boldsymbol{y}f_{a}((1 +z)b +\boldsymbol{p} \cdot \boldsymbol{y} ,\boldsymbol{p} ,r)$. The facts that $af_{a}(a ,\boldsymbol{p} ,r)$ is increasing in $a$ on $\mathbb{R}_{+}$ and $f_{a}$ is decreasing in $a$ imply that $(1 +z)f_{a}((1 +z)b +\boldsymbol{p} \cdot \boldsymbol{y} ,\boldsymbol{p} ,r)$ is increasing in $z$ on $I$. Note that if $f_{a}$ is strictly decreasing, then $(1 +r)f_{a}((1 +r)b +\boldsymbol{p} \cdot \boldsymbol{y} ,\boldsymbol{p} ,r)$ is strictly increasing in $r$.} $b >0$, and from the concavity of $f$ if $b \leq 0$. The second inequality follows from property (i). Thus, the derivative of the function $f((1 +r)b +\boldsymbol{p} \cdot \boldsymbol{y} ,\boldsymbol{p} ,r)$ with respect to $b$ is increasing in $r$. We conclude that $f((1 +r)b +\boldsymbol{p} \cdot \boldsymbol{y} ,\boldsymbol{p} ,r)$ has increasing differences in $(b ,r)$. Thus, the function \begin{equation*}v(a -b ,\boldsymbol{p}) +\beta \sum \limits _{\boldsymbol{y} \in \mathcal{Y}}^{\,}e(\boldsymbol{y})f((1 +r)b +\boldsymbol{p} \cdot \boldsymbol{y} ,\boldsymbol{p} ,r)
\end{equation*} has increasing differences in $(b,r)$ as the sum of functions with increasing differences. Recall that 
$C(a ,\boldsymbol{p},r) =[ -\min _{\boldsymbol{y} \in \mathcal{Y}}\boldsymbol{p} \cdot \boldsymbol{y} /r ,\min \{a ,\sum _{i =1}^{n}p_{i}\overline{b} /(1-r)^{2}\}]$
 is the interval from which an agent may choose his level of savings. Note that $C$ is ascending in $r$ (i.e., $r_{2} \geq r_{1}$, $b \in C(a,\boldsymbol{p},r_{1})$, and $b^{ \prime } \in C(a,\boldsymbol{p},r_{2})$ imply $\max \{b ,b^{ \prime }\} \in C(a,\boldsymbol{p},r_{2})$ and $\min \{b ,b^{ \prime }\} \in C(a,\boldsymbol{p},r_{1})$). Theorem 6.1 in \cite{topkis1978minimizing} implies that \begin{equation*}g^{f}(a ,\boldsymbol{p} ,r) : =\ensuremath{\operatorname*{argmax}}_{b \in C(a ,\boldsymbol{p},r)}v(a -b ,\boldsymbol{p}) +\beta \sum \limits _{\boldsymbol{y} \in \mathcal{Y}}^{\,}e(\boldsymbol{y})f((1 +r)b +\boldsymbol{p} \cdot \boldsymbol{y} ,\boldsymbol{p} ,r)
\end{equation*} is increasing in $r$. The envelope theorem (see \cite{benveniste1979})\ implies that $Tf$ is differentiable and $(Tf)_{a}(a ,\boldsymbol{p} ,r) =v_{a}(a -g^{f}(a ,\boldsymbol{p} ,r) ,\boldsymbol{p})$ when $a -g(a ,\boldsymbol{p} ,r) >0$ (which always holds in our case, because of Assumption \ref{Assumption 0}). 

Using the facts that $v$ has increasing differences in $(a ,b)$ and that $g^{f}(a ,\boldsymbol{p} ,r) \geq g^{f}(a ,\boldsymbol{p} ,r^{ \prime })$ yield
\begin{equation*}(Tf)_{a}(a ,\boldsymbol{p} ,r) =v_{a}(a -g^{f}(a ,\boldsymbol{p} ,r) ,\boldsymbol{p}) \geq v_{a}(a -g^{f}(a ,\boldsymbol{p} ,r^{ \prime }) ,\boldsymbol{p}) =(Tf)_{a}(a ,\boldsymbol{p} ,r^{ \prime }) .
\end{equation*} Thus, $Tf$ has increasing differences in $(a ,r)$. Let $a \geq 0$. We have
\begin{align*}a(Tf)_{a}(a ,\boldsymbol{p} ,r) &  =av_{a}(a -g^{f}(a ,\boldsymbol{p} ,r) ,\boldsymbol{p}) \\
 &  =a\gamma (a -g^{f}(a ,\boldsymbol{p} ,r))^{\gamma  -1}z(\boldsymbol{p}) \\
 &  =\frac{a}{a -g^{f}(a ,\boldsymbol{p} ,r)}\gamma (a -g^{f}(a ,\boldsymbol{p} ,r)^{\gamma }z(\boldsymbol{p}) .\end{align*}
 Since $Tf$ is concave in $a$ (see Lemma \ref{Lemma 1}) for $a \geq a^{ \prime }$ we have \begin{equation*}\gamma (a -g^{f}(a ,\boldsymbol{p} ,r))^{\gamma  -1}z(\boldsymbol{p}) =(Tf)_{a}(a ,\boldsymbol{p} ,r) \leq (Tf)_{a}(a^{ \prime } ,\boldsymbol{p} ,r) =\gamma (a^{ \prime } -g^{f}(a^{ \prime } ,\boldsymbol{p} ,r))^{\gamma  -1}z(\boldsymbol{p})
\end{equation*}
which implies that the function $a -g^{f}(a ,\boldsymbol{p} ,r)$ is increasing in $a$. We conclude that the function $(a -g^{f}(a ,\boldsymbol{p} ,r))^{\gamma }$ is increasing in $a$. Furthermore, the function $\frac{a}{a -g^{f}(a ,\boldsymbol{p} ,r)}$ is increasing in $a$ on $\mathbb{R}_{+}$.\footnote{
To see this, note that $a -g^{f}(a ,\boldsymbol{p} ,r) : =k(a)$ is concave since $g^{f}$ is convex in $a$. Thus, for $a^{ \prime } >a \geq 0$ we have 
\begin{equation*}\frac{k(a^{ \prime }) -k(a)}{a^{ \prime } -a} \leq \frac{k(a^{ \prime }) -k(0)}{a^{ \prime } -0}.
\end{equation*}Rearranging and using the fact that $k(0) >0$ yield \begin{equation*}\frac{a^{ \prime }}{k(a^{ \prime })} \geq \frac{a}{k(a)} .
\end{equation*}\par ~~ } 

Thus, $a(Tf)_{a}(a,\boldsymbol{p},r)$ is increasing on $\mathbb{R}_{+}$ as the product of two positive increasing functions.

Define $f^{n} =T^{n}f : =T(T^{n -1}f)$ for $n =1 ,2 , . . .$ where $T^{0}f : =f$. We conclude that $f^{n}(a ,\boldsymbol{p} ,r)$ is bounded, concave, increasing, and differentiable in $a$ with increasing differences in $(a ,r)$, and that $af_{a}^{n}(a ,\boldsymbol{p} ,r)\ $is increasing in $a$ on $\mathbb{R}_{+}$ for all $n$. The argument above shows that $g^{f_{n}}(a ,\boldsymbol{p} ,r)$ is increasing in $r$ for all $n$. Theorem 3.8 and Theorem 9.9 in \cite{stokey1989} show that $g^{f_{n}}$ converges pointwise to $g$. Thus, the savings policy function $g$ is increasing in $r$. Furthermore, \begin{equation*}\lim _{n \rightarrow \infty }f_{a}^{n}(a ,\boldsymbol{p} ,r) =\lim _{n \rightarrow \infty }\gamma (a -g^{f_{n}}(a ,\boldsymbol{p} ,r))^{\gamma  -1}z(\boldsymbol{p}) =\gamma (a -g(a ,\boldsymbol{p} ,r))^{\gamma  -1}z(\boldsymbol{p}) =V_{a}(a ,\boldsymbol{p} ,r) .
\end{equation*}
Thus, $aV_{a}(a ,\boldsymbol{p} ,r)$ is increasing in $a$ on $\mathbb{R}_{+}$ and has increasing differences in $(a ,r)$. The same argument as the argument above shows that the savings policy function $g$ is increasing in $r$.

\textbf{Step 4. }If $(\boldsymbol{p} ,r)$ and $(\boldsymbol{p} ,r^{ \prime })$ are equilibrium prices with $r >r^{ \prime }$ then $\int g(a ,\boldsymbol{p} ,r)\mu (da;\boldsymbol{p} ,r) >\int g(a ,\boldsymbol{p} ,r^{ \prime })\mu (da;\boldsymbol{p} ,r^{ \prime })$. 

Let $r >r^{ \prime }$. We first show that $g(a ,\boldsymbol{p} ,r) >g(a ,\boldsymbol{p} ,r^{ \prime })$ for all $a \in \widetilde{A}$, and all $\boldsymbol{p} \gg 0$ where $\widetilde{A} =\{a:g(a ,\boldsymbol{p} ,r^{ \prime }) \in  \operatorname{int}C(a,\boldsymbol{p},r^{ \prime })\}$ is the set of wealth levels such that the optimal savings decision is interior. Suppose, in contradiction, that $g(a ,\boldsymbol{p} ,r^{ \prime }) =g(a ,\boldsymbol{p} ,r)$ for some $a \in \widetilde{A}$. Since $V$ is differentiable and strictly concave in $a$ (see Lemma \ref{Lemma 1}), the arguments in Step 3 imply that the function $(1 +r)V_{a}((1 +r)b +\boldsymbol{p} \cdot \boldsymbol{y} ,\boldsymbol{p} ,r)$ is strictly increasing in $r$. The first order condition implies that 
\begin{align*}0 &  = -z(\boldsymbol{p})\gamma (a -g(a ,\boldsymbol{p} ,r^{ \prime }))^{\gamma  -1} +\beta (1 +r^{ \prime })\sum \limits _{\boldsymbol{y} \in \mathcal{Y}}^{\,}e(\boldsymbol{y})V_{a}((1 +r^{ \prime })g(a ,\boldsymbol{p} ,r^{ \prime }) +\boldsymbol{p} \cdot \boldsymbol{y} ,\boldsymbol{p} ,r^{ \prime }) \\
 &  < -z(\boldsymbol{p})\gamma (a -g(a ,\boldsymbol{p} ,r))^{\gamma  -1} +\beta (1 +r)\sum \limits _{\boldsymbol{y} \in \mathcal{Y}}^{\,}e(\boldsymbol{y})V_{a}((1 +r)g(a ,\boldsymbol{p},r) +\boldsymbol{p} \cdot \boldsymbol{y} ,\boldsymbol{p} ,r) \leq 0\end{align*} which is a contradiction. 

For $\widetilde{\lambda }_{1} ,\widetilde{\lambda }_{2} \in \mathcal{P}(\mathbb{R})$ we define the partial order $ \succeq _{I}$ by $\widetilde{\lambda }_{2} \succeq _{I}\widetilde{\lambda }_{1}$ if and only if $\int f(a)\lambda _{2}(da) \geq \int f(a)\lambda _{1}(da)$ for every increasing  function $f$. 

Assume that $\lambda ( \cdot  ,\boldsymbol{p} ,r) \succeq _{I}\lambda ( \cdot  ,\boldsymbol{p} ,r^{ \prime })$.  Then, for every increasing function $f$ we have
\begin{align*}\int f(a)M\lambda (da;\boldsymbol{p} ,r) &  =\int \sum \limits _{\boldsymbol{y} \in \mathcal{Y}}e(\boldsymbol{y})f((1 +r)g(a ,\boldsymbol{p} ,r) +\boldsymbol{p} \cdot \boldsymbol{y})\lambda (da;\boldsymbol{p} ,r) \\
 &  \geq \int \sum \limits _{\boldsymbol{y} \in \mathcal{Y}}e(\boldsymbol{y})f((1 +r^{ \prime })g(a ,\boldsymbol{p} ,r^{ \prime }) +\boldsymbol{p} \cdot \boldsymbol{y})\lambda (da;\boldsymbol{p} ,r) \\
 &  \geq \int \sum \limits _{\boldsymbol{y} \in \mathcal{Y}}e(\boldsymbol{y})f((1 +r^{ \prime })g(a ,\boldsymbol{p},r^{ \prime }) +\boldsymbol{p} \cdot \boldsymbol{y})\lambda (da;\boldsymbol{p} ,r^{ \prime }) \\
 &  =\int f(a)M\lambda (da;\boldsymbol{p} ,r^{ \prime }) .\end{align*}The equalities follow from Equation (\ref{eq:stationary}) (see Lemma \ref{Unique stationary is contin}). The first inequality follows from the fact that $g$ is increasing in $r$. The second inequality follows from the facts that $g$ is increasing in $a$ and $\lambda ( \cdot ;\boldsymbol{p} ,r) \succeq _{I}\lambda ( \cdot ;\boldsymbol{p} ,r^{ \prime })$. We conclude that $M^{k}\lambda ( \cdot ;\boldsymbol{p} ,r) \succeq _{I}M^{k}\lambda ( \cdot ;\boldsymbol{p} ,r^{ \prime })$ for all $k =1 ,2 ,...$. From Lemma \ref{unique stationary dist}, the sequence $\{M^{k}\lambda \}$ converges weakly to $\mu $ for all $(\boldsymbol{p} ,r)$. Since $ \succeq _{I}$ is closed under weak convergence, we conclude that  $\mu ( \cdot ;\boldsymbol{p} ,r) \succeq _{I}\mu ( \cdot ;\boldsymbol{p} ,r^{ \prime })$. 

 Suppose that $(\boldsymbol{p} ,r)$ and $(\boldsymbol{p} ,r^{ \prime })$ are equilibrium prices with $r >r^{ \prime }$. We have
\begin{equation*}\int g(a ,\boldsymbol{p} ,r)\mu (da;\boldsymbol{p} ,r) >\int g(a ,\boldsymbol{p} ,r^{ \prime })\mu (da;\boldsymbol{p} ,r) \geq \int g(a ,\boldsymbol{p} ,r^{ \prime })\mu (da;\boldsymbol{p} ,r^{ \prime }) .
\end{equation*}The first inequality follows from the fact that $g$ is strictly increasing in $r$ on $\widetilde{A}$ (and we have $\mu (\widetilde{A};\boldsymbol{p} ,r) >0$ since $(\boldsymbol{p} ,r)$ are equilibrium prices). The second inequality follows from the facts that $g$ is increasing in $a$ and $\mu ( \cdot ;\boldsymbol{p} ,r) \succeq _{I}\mu ( \cdot ;\boldsymbol{p} ,r^{ \prime })$. 

     \textbf{Step 5. }Suppose that $(\boldsymbol{p} ,r)$ and $(\boldsymbol{p}^{ \prime } ,r^{ \prime })$ are equilibrium prices. From Step 1, we know that $\boldsymbol{p}^{ \prime } =\boldsymbol{p}$. From Step 4, if $r >r^{ \prime }$ then $0 =\int g(a ,\boldsymbol{p} ,r)\mu (da;\boldsymbol{p} ,r) >\int g(a ,\boldsymbol{p} ,r^{ \prime })\mu (da;\boldsymbol{p} ,r^{ \prime }) =0$ which is a contradiction. We conclude that $(\boldsymbol{p} ,r)$$ =(\boldsymbol{p}^{ \prime } ,r^{ \prime })$. Thus, there is at most one CSE. It easy to see that Assumptions \ref{Assumption 0} is satisfied so Theorem \ref{Theorem exist} implies that there exists at least one CSE. We conclude that there is a unique CSE.    
\end{proof}

\subsection{Proof of Theorem \ref{Theorem wealth ineq}}

In this section we prove Theorem \ref{Theorem wealth ineq}.

\noindent \textbf{Theorem \ref{Theorem wealth ineq}}
\emph{Assume that $U(\boldsymbol{x}) =\sum _{i =1}^{n}\alpha _{i}x_{i}^{\gamma }$ for some $0 <\gamma  <1$, $\alpha _{i} >0$, $\sum _{i =1}^{n}\alpha _{i} =1$. Assume that the endowments process $e$ is riskier than the endowments process $e^{ \prime }$. Then \newline
(i) The partial equilibrium wealth inequality is higher under $e$ than under $e^{ \prime }$, i.e., $\mu ( \cdot ;\boldsymbol{p},r,e) \succeq _{I-CX}\mu ( \cdot ;\boldsymbol{p} ,r ,e^{ \prime })$ for all $(\boldsymbol{p} ,r) \in \boldsymbol{P}$. In addition, if  $(\boldsymbol{p}(e),r(e))$ are equilibrium prices under the endowments process $e$ then $\mu ( \cdot ;\boldsymbol{p}(e),r(e),e) \succeq _{CX}\mu ( \cdot ;\boldsymbol{p}(e) ,r(e) ,e^{ \prime })$.    \newline (ii) The equilibrium prices of goods do not change, i.e., $\boldsymbol{p}(e) =\boldsymbol{p}(e^{ \prime })$. The equilibrium interest rate is lower under $q$ than under $e^{ \prime }$, i.e., $r(e^{ \prime }) \geq r(e)$.}

\begin{proof}
(i) Fix $(\boldsymbol{p} ,r) \in \boldsymbol{P}$. Assume that the endowments process $e$ is riskier than the endowments process $e^{ \prime }$. From Theorem 2 in \cite{light2017precautionary}, we can show that $g(a ,\boldsymbol{p} ,r ,e) \geq g(a ,\boldsymbol{p} ,r ,e^{ \prime })$ for all $(a ,\boldsymbol{p} ,r)$.     

Suppose that $\lambda ( \cdot ;\boldsymbol{p} ,r ,e) \succeq _{I -CX}\lambda ( \cdot ;\boldsymbol{p} ,r ,e^{\prime })$.  Then for every convex and increasing function $f$ we have
\begin{align*}\int f(a)M\lambda (da;\boldsymbol{p} ,r , e) &  =\int \sum \limits _{\boldsymbol{y} \in \mathcal{Y}}^{\,}e(\boldsymbol{y})f((1 +r)g(a ,\boldsymbol{p} ,r ,e) +\boldsymbol{p} \cdot \boldsymbol{y})\lambda (da;\boldsymbol{p} ,r , e) \\
 &  \geq \int \sum \limits _{\boldsymbol{y} \in \mathcal{Y}}^{\,}e^{ \prime }(\boldsymbol{y})f((1 +r)g(a ,\boldsymbol{p} ,r ,e) +\boldsymbol{p} \cdot \boldsymbol{y})\lambda (da;\boldsymbol{p} ,r , e) \\
 &  \geq \int \sum \limits _{\boldsymbol{y} \in \mathcal{Y}}^{\,}e^{ \prime }(\boldsymbol{y})f((1 +r)g(a ,\boldsymbol{p} ,r ,e^{ \prime }) +\boldsymbol{p} \cdot \boldsymbol{y})\lambda (da;\boldsymbol{p} ,r , e) \\
 &  \geq \int \sum \limits _{\boldsymbol{y} \in \mathcal{Y}}^{\,}e^{ \prime }(\boldsymbol{y})f((1 +r)g(a ,\boldsymbol{p} ,r ,e^{ \prime }) +\boldsymbol{p} \cdot \boldsymbol{y})\lambda (da;\boldsymbol{p} ,r ,e^{ \prime }) \\
 &  =\int f(a)M\lambda (da;p ,r ,e^{ \prime }) .\end{align*}The equalities follow from Equation (\ref{eq:stationary}) (see Lemma \ref{Unique stationary is contin}). The first inequality follows from the fact that $f((1 +r)g(a ,\boldsymbol{p} ,r ,e) +\boldsymbol{p} \cdot \boldsymbol{y})$ is convex in $\boldsymbol{y}$ as the composition of a convex and increasing function with a convex function. The second inequality follows from the facts that $g(a ,\boldsymbol{p} ,r ,e) \geq g(a ,\boldsymbol{p} ,r ,e^{ \prime })$ and $f$ is increasing. The third inequality follows from the fact that $g$ is convex and increasing in $a$ (see Step 2 in the proof of Theorem \ref{Theorem uniq}), which implies that $f((1 +r)g(a ,\boldsymbol{p} ,r) +\boldsymbol{p} \cdot \boldsymbol{y})$ is convex and increasing in $a$, and from the fact that $\lambda ( \cdot ;\boldsymbol{p} ,r , e) \succeq _{I -CX}\lambda ( \cdot ;\boldsymbol{p} ,r ,e^{ \prime })$. 
 
 We conclude that $M^{k}\lambda ( \cdot ;\boldsymbol{p} ,r , e) \succeq _{I -CX}M^{k}\lambda ( \cdot ;\boldsymbol{p} ,r ,e^{ \prime })$ for all $k =1 ,2 ,...$. From Lemma \ref{unique stationary dist}, the sequence $\{M^{k}\lambda \}$ converges weakly to $\mu $ for all $(\boldsymbol{p} ,r)$. Since under our assumptions (see Theorem 1.5.9 in \cite{muller2002comparison})  $ \succeq _{I -CX}$ is closed under weak convergence, we conclude that $\mu ( \cdot ;\boldsymbol{p} ,r ,e) \succeq _{I -CX}\mu ( \cdot ;\boldsymbol{p} ,r ,e^{ \prime })$.
 
 Now assume that $(\boldsymbol{p}(e),r(e))$ are equilibrium prices under the endowment process $e$, so 
 \begin{equation*} 
 \int g(a ,\boldsymbol{p}(e) ,r(e))\mu (da;\boldsymbol{p}(e),r(e), e)=0. 
 \end{equation*}

 $e \succeq _{CX}e ^{ \prime }$ and the linearity of the function $\boldsymbol{p} \cdot \boldsymbol{y}$ imply that $\sum e(\boldsymbol{y})\boldsymbol{p} \cdot \boldsymbol{y} =\sum e^{ \prime }(\boldsymbol{y})\boldsymbol{p} \cdot \boldsymbol{y}$. We have 
\begin{align*}\int a\mu (da;\boldsymbol{p}(e) ,r(e) ,e) &  =\sum \limits _{\boldsymbol{y} \in \mathcal{Y}}^{\,}e(\boldsymbol{y})((1 +r(e))\int g(a ,\boldsymbol{p}(e) ,r(e))\mu (da;\boldsymbol{p}(e) ,r(e) ,e) +\boldsymbol{p}(e) \cdot \boldsymbol{y}) \\
 &  =\sum \limits _{\boldsymbol{y} \in \mathcal{Y}}^{\,}e(\boldsymbol{y})\boldsymbol{p}(e) \cdot \boldsymbol{y} \\
 &  =\sum _{\boldsymbol{y} \in \mathcal{Y}}e^{ \prime }(\boldsymbol{y})\boldsymbol{p}(e) \cdot \boldsymbol{y} \\
 &  =\int a\mu (da;\boldsymbol{p}(e) ,r(e) ,e^{ \prime }) .\end{align*} We proved that $\mu ( \cdot ;\boldsymbol{p}(e) ,r(e) ,e) \succeq _{I -CX}\mu ( \cdot ;\boldsymbol{p}(e) ,r(e) ,e^{ \prime })$ and
 
 $\int a\mu (da;\boldsymbol{p}(e) ,r(e) ,e) =\int a\mu (da;\boldsymbol{p}(e) ,r(e) ,e^{ \prime })$. 
 
 This implies that $\mu ( \cdot ;\boldsymbol{p}(e) ,r(e) ,e) \succeq _{CX}\mu ( \cdot ;\boldsymbol{p}(e) ,r(e) ,e^{ \prime })$ (see Theorem 1.5.3 in \cite{muller2002comparison}).

(ii) Assume that $(\boldsymbol{p}(e) ,r(e)$) and $(\boldsymbol{p}(e^{ \prime }) ,r(e^{ \prime }))$ are equilibrium prices.  Suppose, in contradiction, that $\boldsymbol{p}(e^{ \prime }) \neq \boldsymbol{p}(e)$. From Proposition \ref{Prop homogenous of } we can normalize the prices and set $\boldsymbol{p}(e) \geq \boldsymbol{p}(e^{ \prime })$ and $p_{k}^{ \prime } =p_{k} =1$ for some $1 \leq k \leq n$. 

We have $\boldsymbol{x}^{ \ast }(a -g(a ,\boldsymbol{p} ,r ,e) ,\boldsymbol{p}) =(z_{1}(\boldsymbol{p})(a -g(a ,\boldsymbol{p} ,r ,e)) , . . . ,z_{n}(\boldsymbol{p})(a -g(a ,\boldsymbol{p} ,r ,e))$ where $z_{i}(\boldsymbol{p})$ is a positive function and $z_{1}(\boldsymbol{p}) \geq z_{1}(\boldsymbol{p}^{ \prime })$ (see Step 1 of the Proof of Theorem \ref{Theorem uniq}). 

Since $\int a\mu (da;\boldsymbol{p}(e) ,r(e) ,e) =\int a\mu (da;\boldsymbol{p}(e) ,r(e) ,e^{ \prime })$ we have 
\begin{align*}\int x_{k}^{ \ast }(a -g(a ,\boldsymbol{p}(e) ,r(e) ,e)) ,\boldsymbol{p})\mu (da;\boldsymbol{p}(e) ,r(e) ,e) &  =z_{k}(\boldsymbol{p}(e))\int (a -g(a ,\boldsymbol{p}(e),r(e),e))\mu (da;\boldsymbol{p}(e) ,r(e) ,e) \\
 &  =z_{k}(\boldsymbol{p}(e))\int a\mu (da;\boldsymbol{p}(e) ,r(e) ,e) \\
 &  =z_{k}(\boldsymbol{p}(e))\int a\mu (da;\boldsymbol{p}(e) ,r(e) ,e^{ \prime }) \\
 &  >z_{k}(\boldsymbol{p}(e^{ \prime }))\int a\mu (da;\boldsymbol{p}(e^{ \prime }) ,r(e^{ \prime }) ,e^{ \prime }) \\
 &  =\int x_{k}^{ \ast }(a -g(a ,\boldsymbol{p}(e^{ \prime }) ,r(e^{ \prime }),e^{ \prime }))\mu (da;\boldsymbol{p}(e^{ \prime }) ,r(e^{ \prime }) ,e^{ \prime }) .\end{align*}
 The inequality follows from the same argument as in Step 1 of the proof of Theorem \ref{Theorem uniq}. Since $e \succeq _{CX}e^{ \prime }$, we have $e_{i} \succeq _{CX}e_{i}^{ \prime }$ for all $1 \leq i \leq n$  (see Theorem 3.4.4. In \cite{muller2002comparison}). Recall that $e_{i} \succeq _{CX}e_{i}^{ \prime }$ implies that $\sum e_{i}^{ \prime }(y_{i})y_{i} =\sum e_{i}(y_{i})y_{i}$. Thus, $0 =\zeta _{k}(\boldsymbol{p}(e) ,r(e) ,e) >\zeta _{k}(\boldsymbol{p}(e^{ \prime }) ,r(e^{ \prime }) ,e^{ \prime }) =0$ which is a contradiction. We conclude that $\boldsymbol{p}(e) =\boldsymbol{p}(e^{ \prime })$. 

Now assume, in contradiction, that $r(e) >r(e^{ \prime })$. We have 
\begin{align*}0 =\int g(a ,\boldsymbol{p}(e) ,r(e) ,e)\mu (da;\boldsymbol{p}(e) ,r(e) ,e) &  >\int g(a ,\boldsymbol{p}(e) ,r(e^{ \prime }) ,e)\mu (da;\boldsymbol{p}(e) ,r(e^{ \prime }) ,e) \\
 &  \geq \int g(a ,\boldsymbol{p}(e) ,r(e^{ \prime }) ,e)\mu (da;\boldsymbol{p}(e) ,r(e^{ \prime }) ,e^{ \prime }) \\
 &  \geq \int g(a ,\boldsymbol{p}(e^{ \prime }) ,r(e^{ \prime }) ,e^{ \prime })\mu (da;\boldsymbol{p}(e^{ \prime }) ,r(e^{ \prime }) ,e^{ \prime }) =0\end{align*}which is a contradiction. The first (strict) inequality follows from Step 4 of the proof of Theorem \ref{Theorem uniq}. The second inequality follows from the facts that $g$ is convex in $a$ and $\mu ( \cdot ;\boldsymbol{p} ,r ,e) \succeq _{CX}\mu ( \cdot ;\boldsymbol{p} ,r ,e^{ \prime })$. The third inequality follows from the facts that $g(a ,\boldsymbol{p} ,r ,e) \geq g(a ,\boldsymbol{p} ,r ,e^{ \prime })$ and $\boldsymbol{p}(e) =\boldsymbol{p}(e^{ \prime })$.  We conclude that $r(e^{ \prime }) \geq r(e)$.
\end{proof}

\bibliographystyle{ecta}
\bibliography{ArrowDeb}

\end{document}